\documentclass[journal]{IEEEtran}
\usepackage[cmex10]{amsmath}
\usepackage{amssymb,amsthm}

\newcommand{\bbC}{\mathbb{C}}
\newcommand{\bbF}{\mathbb{F}}
\newcommand{\bbH}{\mathbb{H}}
\newcommand{\bbR}{\mathbb{R}}

\newcommand{\bfA}{\boldsymbol{A}}
\newcommand{\bfB}{\boldsymbol{B}}
\newcommand{\bfC}{\boldsymbol{C}}
\newcommand{\bfD}{\boldsymbol{D}}
\newcommand{\bfE}{\boldsymbol{E}}
\newcommand{\bfG}{\boldsymbol{G}}
\newcommand{\bfP}{\boldsymbol{P}}
\newcommand{\bfU}{\boldsymbol{U}}
\newcommand{\bfV}{\boldsymbol{V}}
\newcommand{\bfx}{\boldsymbol{x}}
\newcommand{\bfy}{\boldsymbol{y}}
\newcommand{\bfZ}{\boldsymbol{Z}}

\newcommand{\bfone}{\boldsymbol{1}}
\newcommand{\bfzero}{\boldsymbol{0}}

\newcommand{\bfdelta}{\boldsymbol{\delta}}
\newcommand{\bfPi}{\boldsymbol{\varPi}}
\newcommand{\bfupsilon}{\boldsymbol{\upsilon}}
\newcommand{\bfUpsilon}{\boldsymbol{\varUpsilon}}
\newcommand{\bfphi}{\boldsymbol{\varphi}}
\newcommand{\bfPhi}{\boldsymbol{\varPhi}}
\newcommand{\bfPsi}{\boldsymbol{\varPsi}}
\newcommand{\bfOmega}{\boldsymbol{\varOmega}}

\newcommand{\brI}{\mathbf{I}}
\newcommand{\brJ}{\mathbf{J}}
\newcommand{\brM}{\mathbf{M}}
\newcommand{\brR}{\mathbf{R}}
\newcommand{\brT}{\mathbf{T}}

\newcommand{\calE}{\mathcal{E}}
\newcommand{\calM}{\mathcal{M}}
\newcommand{\calN}{\mathcal{N}}
\newcommand{\calS}{\mathcal{S}}
\newcommand{\calU}{\mathcal{U}}

\newcommand{\rmA}{\mathrm{A}}
\newcommand{\rmi}{\mathrm{i}}
\newcommand{\rmj}{\mathrm{j}}
\newcommand{\rmk}{\mathrm{k}}
\newcommand{\rms}{\mathrm{s}}
\newcommand{\rmS}{\mathrm{S}}

\newcommand{\im}{\operatorname{im}}
\newcommand{\op}{{\operatorname{op}}}
\newcommand{\SA}{\operatorname{SA}}
\newcommand{\Tr}{\operatorname{Tr}}
\newcommand{\Fro}{{\operatorname{Fro}}}
\newcommand{\dist}{\operatorname{dist}}
\newcommand{\rank}{\operatorname{rank}}
\newcommand{\real}{\operatorname{Re}}
\newcommand{\Span}{\operatorname{span}}
\newcommand{\EITFF}{\operatorname{EITFF}}

\newcommand{\bigparen}[1]{\bigl({#1}\bigr)}

\newcommand{\bigbracket}[1]{\bigl[{#1}\bigr]}
\newcommand{\set}[1]{\{{#1}\}}
\newcommand{\norm}[1]{\|{#1}\|}

\newcommand{\ip}[2]{\langle{#1},{#2}\rangle}

\newtheorem{theorem}{Theorem}[section]
\newtheorem{lemma}[theorem]{Lemma}
\theoremstyle{definition}
\newtheorem{remark}[theorem]{Remark}
\newtheorem{example}[theorem]{Example}

\begin{document}
\title{Radon--Hurwitz Grassmannian codes}

\author{Matthew~Fickus,~\IEEEmembership{Senior Member,~IEEE}, Enrique Gomez-Leos, Joseph W.\ Iverson%
\thanks{M.~Fickus is with the Department of Mathematics and Statistics, Air Force Institute of Technology, Wright-Patterson Air Force Base, OH 45433, USA, e-mail: Matthew.Fickus@afit.edu, Matthew.Fickus@us.af.mil.}%
\thanks{E.~Gomez-Leos and J.~W.~Iverson are with the Department of Mathematics, Iowa State University, Ames, IA 50011, USA.}}

\maketitle


\begin{abstract}
Every equi-isoclinic tight fusion frame (EITFF) is a type of optimal code in a Grassmannian,
consisting of subspaces of a finite-dimensional Hilbert space for which the smallest principal angle between any pair of them is as large as possible.
EITFFs yield dictionaries with minimal block coherence and so are ideal for certain types of compressed sensing.
By refining classical work of Lemmens and Seidel based on Radon--Hurwitz theory,
we fully characterize EITFFs in the special case where the dimension of the subspaces is exactly one-half of that of the ambient space.
We moreover show that each such ``Radon--Hurwitz EITFF" is highly symmetric, where every even permutation is an automorphism.
\end{abstract}


\section{Introduction}
\renewcommand{\thesection}{\arabic{section}}

Let $\bbF$ be either $\bbR$ or $\bbC$,
and let $d\geq r\geq 1$ and $n\geq 2$ be integers.
For each index $i\in[n]:=\set{1,2,\dotsc,n}$,
let $\calU_i$ be a subspace of $\bbF^d$ of dimension $r$,
and let $\bfPhi_i$ be a corresponding \textit{isometry}, that is,
a $d\times r$ matrix whose columns form an orthonormal basis for $\calU_i$.
The \textit{block coherence} of $(\bfPhi_i)_{i=1}^n$ is
\begin{equation}
\label{eq.block coherence}
\smash{\mu
:=\max_{i\neq j}\norm{\bfPhi_{i}^*\bfPhi_{j}^{}}_{\op}.}
\end{equation}
Isometries $(\bfPhi_i)_{i=1}^n$ with small block coherence are useful for certain types of compressed sensing.
For example,
a sequence $(\bfx_i)_{i=1}^n$ of $n$ vectors in $\bbF^r$ can be uniquely and efficiently recovered from \smash{$\bfy=\sum_{i=1}^n\bfPhi_i\bfx_i\in\bbF^d$} using \textit{block orthogonal matching pursuit}~\cite{Tropp04,EldarKB10,CalderbankTX15}, provided
\begin{equation*}
\smash{\#\set{i\in[n]: \bfx_i\neq\bfzero}<\tfrac12(\tfrac1{\mu}+1).}
\end{equation*}
That said, it is well known~\cite{LemmensS73b,Welch74,ConwayHS96,StrohmerH03,DhillonHST08,CalderbankTX15} that the block coherence of any such sequence of isometries is bounded below by the \textit{(spectral) Welch bound}:
\begin{equation}
\label{eq.Welch bound}
\mu^2\geq\tfrac{nr-d}{d(n-1)}.
\end{equation}
Moreover,
as detailed in the next section,
equality holds in~\eqref{eq.Welch bound} if and only if $(\calU_i)_{i=1}^{n}$ is an \textit{equi-isoclinic tight fusion frame} for $\bbF^d$,
denoted here as an $\EITFF_\bbF(d,r,n)$.
In particular, every such EITFF is an $n$-element optimal code in the Grassmannian (manifold) that consists of all $r$-dimensional subspaces of $\bbF^d$,
with respect to the \textit{spectral distance} of~\cite{DhillonHST08}:
achieving equality in~\eqref{eq.Welch bound} guarantees that the smallest \textit{principal angle} between any pair of its subspaces is as large as possible.

These facts naturally lead to an existence problem:
for what such $\bbF$, $d$, $r$ and $n$ does an $\EITFF_\bbF(d,r,n)$ exist?
This problem remains mostly open despite decades of incremental progress since the seminal work of Lemmens and Seidel~\cite{LemmensS73b}.
In particular, EITFFs with $r=1$ equate to \textit{equiangular tight frames} (ETFs),
which remain poorly understood, particularly in the complex setting;
see~\cite{FickusM16} for a survey of them.
Famously, (a weak version of) \textit{Zauner's conjecture} asserts that an $\EITFF_\bbC(d,1,d^2)$ exists for all positive integers $d$~\cite{Zauner99},
and it remains unresolved~\cite{FuchsHS17}.
With the exception of those that arise from the quaternionic ETFs of~\cite{CohnKM16} via Hoggar's method~\cite{Hoggar77},
all known EITFFs have been obtained via explicit construction.

In~\cite{LemmensS73b}, Lemmens and Seidel extended work by Wong~\cite{Wong61} and Wolf~\cite{Wolf63} to fully characterize the existence of $\EITFF_\bbF(d,r,n)$ with $d=2r$ and $\bbF=\bbR$.
However, their classical proof techniques are nonideal from the perspective of modern compressed sensing applications as they do not immediately yield explicit isometries $(\bfPhi_i)_{i=1}^{n}$ for such $\EITFF$s.
In this paper,
we remedy this issue,
and then explore some of the additional benefits of doing so.
Specifically, in Theorem~\ref{thm.RH EITFF isometries},
we provide a new explicit characterization of the isometries of every $\EITFF_\bbF(d,r,n)$ with $d=2r$,
up to a standard notion of equivalence.
In Theorem~\ref{thm.RH EITFF existence},
we then exploit this new characterization to both recover the aforementioned result from~\cite{LemmensS73b} and moreover generalize it to the complex setting,
showing that an $\EITFF_\bbF(2r,r,n)$ exists if and only if $n\leq\rho_\bbF(r)+2$ where $\rho_\bbF(r)$ is the corresponding classical \textit{Radon--Hurwitz number}~\cite{Radon22,Hurwitz23,Adams62,AdamsLP65},
namely
\begin{equation}
\label{eq.RH number}
\rho_\bbF(r)=\left\{\begin{array}{ll}
8b+2^c, &\ \bbF=\bbR,\\
8b+2c+2,&\ \bbF=\bbC,
\end{array}\right.
\end{equation}
where $r=(2a+1)2^{4b+c}$ for some nonnegative integers $a,b,c$ with $c\leq3$.
Results of this type help gauge the potential benefits to sensing efficiency afforded by permitting complex measurements,
as is commonly done in certain applications of compressed sensing such as radar~\cite{HermanS09}.
Here in particular, we see that doing so offers a marginal benefit when $c\neq 3$,
permitting complex such EITFFs with either one or two more subspaces than their real cousins when either $c=0$ or $c\in\set{1,2}$, respectively.
See Section~III for details.

In Section~IV,
we further exploit Theorem~\ref{thm.RH EITFF isometries} to show that any $\EITFF_\bbF(d,r,n)$ with $d=2r$ has alternating or total symmetry (Theorem~\ref{thm.total symmetry suf}) and, in fact, that the latter occurs for many such EITFFs (Theorems~\ref{thm.total symmetry complex} and~\ref{thm.total symmetry real}).
This yields the second-ever construction  of an infinite family of nontrivial EITFFs with total symmetry,
following that of~\cite{FickusIJM22,FickusIJM24}.
Results such as these in which ``optimality implies symmetry" serve as partial converses to some ``symmetry implies optimality" results from the recent literature~\cite{IversonM22,IversonM24};
see~\cite{CoxKMP20} for some other instances of the former.
This work is in part motivated by a ``holy grail" of compressed sensing,
namely the open problem of finding deterministic constructions of (block) sensing matrices which certifiably satisfy the \textit{(block) restricted isometry property}~\cite{EldarM09} to a degree that rivals or exceeds those whose existence is guaranteed by random matrix theory~\cite{BandeiraFMW13}.
Essentially, the goal here is to design $(\calU_i)_{i=1}^{n}$ so that any $k$ of these subspaces are ``almost orthogonal," where the integer $k\geq 2$ is as large as reasonably possible.
Since there are $\binom{n}{k}$ such subcollections,
it is hard to certify that a general $(\calU_i)_{i=1}^{n}$ has this property~\cite{BandeiraDMS13}.
Moreover, though any such $(\calU_i)_{i=1}^{n}$ necessarily has small block coherence, the latter alone is not sufficient~\cite{BandeiraFMW13,FickusJKM18}.
That said, when $(\calU_i)_{i=1}^{n}$ is an EITFF with alternating or total symmetry,
its block coherence is minimal, and for $k \leq n-2$ any two $k$-member subcollections of it are unitarily equivalent, and thus also ``almost orthogonal" to an equal degree.

To be clear, the practical implications of these ideas are limited when applied to the $\EITFF_\bbF(d,r,n)$ that we consider here since they have $d=2r$,
implying that any three of their subspaces are linearly dependent.
Nevertheless,
they inform future discussion of EITFFs in general in regards to the form of their isometries,
and whether they might be real or have total symmetry.
Our formal treatment begins in the next section,
where we discuss some relevant background material.

\renewcommand{\thesection}{\Roman{section}}
\section{Preliminaries}
\renewcommand{\thesection}{\arabic{section}}

In this paper, a \textit{space} is a finite-dimensional Hilbert space over $\bbF$ whose inner product $\ip{\cdot}{\cdot}$ is conjugate-linear in its first argument and is linear in its second argument.
For example,
let $\calN$ be a finite nonempty set,
and let $\bbF^\calN$ be the standard vector space of functions from $\calN$ into $\bbF$ where $\ip{\bfx}{\bfy}:=\sum_{i\in\calN}\overline{\bfx(i)}\bfy(i)$.
Let $(\bfdelta_i)_{i\in\calN}$ be the standard basis in $\bbF^\calN$,
so that $\bfdelta_i(j)$ is $1$ if $i=j$ and $0$ otherwise.

We denote the image (range) and adjoint of a linear map $\bfA$ between spaces by $\im(\bfA)$ and $\bfA^*$, respectively.
For any finite nonempty sets $\calM$ and $\calN$,
we abuse notation and equate a linear map $\bfA:\bbF^\calN\rightarrow\bbF^\calM$ with a matrix $\bfA\in\bbF^{\calM\times\calN}$ as usual,
having $\bfA(i,j)=\ip{\bfdelta_i}{\bfA\bfdelta_j}$ for any $i\in\calM$, $j\in\calN$.
The adjoint of $\bfA\in\bbF^{\calM\times\calN}$ thus equates to its conjugate transpose $\bfA^*\in\bbF^{\calN\times\calM}$,
$\bfA^*(j,i):=\overline{\bfA(i,j)}$.
For integers $m,n\geq 1$, we abbreviate $\bbF^{[n]}$ and $\bbF^{[m]\times[n]}$ as $\bbF^n$ and $\bbF^{m\times n}$, respectively.

The Frobenius inner product of two linear maps from one given space to another is $\ip{\bfA}{\bfB}_\Fro:=\Tr(\bfA^*\bfB^{})$.
We will employ both the associated norm $\norm{\bfA}_{\Fro}:=[\Tr(\bfA^*\bfA)]^{\frac12}$
and the induced operator norm
$\norm{\bfA}_{\op}:=\sup\set{\norm{\bfA\bfx}: \norm{\bfx}\leq 1}$.

The \textit{synthesis map} of a finite sequence $(\bfphi_i)_{i\in\calN}$ of vectors in a space $\calE$ is $\bfPhi:\bbF^\calN\rightarrow\calE$, $\bfPhi\bfx:=\sum_{i\in\calN}\bfx(i)\bfphi_i$.
It is linear with $\im(\bfPhi)=\Span(\bfphi_i)_{i\in\calN}$.
Its adjoint $\bfPhi^*:\calE\rightarrow\bbF^\calN$ is the associated \textit{analysis map}, and is given by $(\bfPhi^*\bfy)(i)=\ip{\bfphi_i}{\bfy}$.
Composing $\bfPhi$ and $\bfPhi^*$ gives the associated \textit{frame operator} $\bfPhi\bfPhi^*:\calE\rightarrow\calE$ given by $\bfPhi\bfPhi^*\bfy=\sum_{i\in\calN}\ip{\bfphi_i}{\bfy}\bfphi_i$
and \textit{Gram matrix} $\bfPhi^*\bfPhi\in\bbF^{\calN\times\calN}$ given by
$(\bfPhi^*\bfPhi)(i,j)=\ip{\bfphi_i}{\bfphi_j}$.
Both are positive semidefinite of rank $\dim(\Span(\bfphi_i)_{i=1}^{n})$.
When $\calE=\bbF^m$ and $\calN=[n]$ for some integers $m,n\geq1$,
$\bfPhi$ is the $m\times n$ matrix whose $i$th column is $\bfphi_i$.
For general $\calE$, we can identify each $\bfphi_i$ with the map $x\mapsto x\bfphi_i$ from $\bbF$ into $\calE$ whose adjoint $\bfphi_i^*$ is the linear functional $\bfy\mapsto\ip{\bfphi_i}{\bfy}$,
implying $\bfPhi\bfPhi^*=\sum_{i\in\calN}\bfphi_i^{}\bfphi_i^*$.

Any linear map $\bfPhi:\bbF^\calN\rightarrow\calE$ is a synthesis map,
namely that of $(\bfphi_i)_{i\in\calN}=(\bfPhi\bfdelta_i)_{i\in\calN}$.
By spectral decomposition, any positive semidefinite $\bfG\in\bbF^{\calN\times\calN}$ is the Gram matrix $\bfPhi^*\bfPhi$ of a sequence $(\bfphi_i)_{i\in\calN}$ in a space of dimension $\rank(\bfG)$ that is unique up to unitaries.
This sequence can be chosen to lie in $\bbF^{\rank(\bfG)}$, if so desired.

\subsection{Equi-isoclinic subspaces}

Let $(\calU_i)_{i=1}^{n}$ be a sequence of $r$-dimensional subspaces of a space $\calE$ of dimension $d$.
For each $i\in[n]$,
let $\bfPhi_i:\bbF^r\mapsto\calE$ be an \textit{isometry} onto $\calU_i$,
that is, be the synthesis map of an orthonormal basis $(\bfphi_{i,k})_{k=1}^{r}$ for $\calU_i$, or equivalently, satisfy $\bfPhi_i^*\bfPhi_i^{}=\brI$ and $\im(\bfPhi_i)=\calU_i$.
For any  $i\in[n]$, the (orthogonal) projection (operator) onto $\calU_i$ is $\bfPi_i:=\bfPhi_i^{}\bfPhi_i^*$.
Meanwhile, for any $i,j \in [n]$, the \textit{cross-Gram matrix} $\bfPhi_i^*\bfPhi_j^{}\in\bbF^{r\times r}$ satisfies
\begin{equation*}
(\bfPhi_i^*\bfPhi_j^{})(k,l)
=\ip{\bfphi_{i,k}}{\bfphi_{j,l}}
\end{equation*}
for any $k,l\in[r]$.
Such matrices are involved in one of the many ways to understand equivalence of subspace sequences, summarized below.

\begin{lemma}[Folklore]
\label{lem.equivalence}
Suppose $(\calU_i)_{i=1}^{n}$ and $(\widehat{\calU}_i)_{i=1}^{n}$ are sequences of \mbox{$r$-dimensional} subspaces of some $d$-dimensional spaces $\calE$ and $\widehat{\calE}$, respectively.
Let $(\bfPi_i)_{i=1}^{n}$ and $(\widehat{\bfPi}_i)_{i=1}^{n}$ be the associated projections, and $(\bfPhi_i)_{i=1}^{n}$ and $(\widehat{\bfPhi}_i)_{i=1}^{n}$ be associated isometries with domain $\bbF^r$.
Then the following are equivalent:
\begin{enumerate}
\renewcommand{\labelenumi}{(\roman{enumi})}
\item
There is a unitary $\bfUpsilon:\calE\rightarrow\widehat{\calE}$ such that
$\widehat{\calU}_i=\bfUpsilon(\calU_i)$ for all $i\in[n]$.

\item
There is a unitary $\bfUpsilon:\calE\rightarrow\widehat{\calE}$ such that
$\widehat{\bfPi}_i=\bfUpsilon\bfPi_i\bfUpsilon^*$ for all $i\in[n]$.
\item
There exist unitaries $(\bfZ_i)_{i=1}^{n}$ in $\bbF^{r\times r}$ and  $\bfUpsilon:\calE\rightarrow\widehat{\calE}$ such that
$\widehat{\bfPhi}_i=\bfUpsilon\bfPhi_i\bfZ_i$ for all $i\in[n]$.

\item
There exist unitaries $(\bfZ_i)_{i=1}^{n}$ in $\bbF^{r\times r}$ such that
\mbox{$\widehat{\bfPhi}_i^*\widehat{\bfPhi}_j^{}=\bfZ_i^*\bfPhi_i^*\bfPhi_j^{}\bfZ_j^{}$}
for all $i,j\in[n]$.
\end{enumerate}
\end{lemma}

For a proof, see the appendix.
We say that two sequences of subspaces $(\calU_i)_{i=1}^{n}$ and $(\widehat{\calU}_i)_{i=1}^{n}$ are \textit{equivalent} if they satisfy (i)--(iv) of Lemma~\ref{lem.equivalence}.
Though such unitaries $\bfUpsilon$ and $(\bfZ_i)_{i=1}^{n}$ are not unique in general,
if (iii) holds for a particular $\bfUpsilon$ then each corresponding $\bfZ_i$ is necessarily $\bfZ_i=\bfPhi_i^*\bfUpsilon^*\widehat{\bfPhi}_i^{}$.
Here, for any sequence $(\calU_i)_{i=1}^{n}$ of subspaces of $\calE$,
taking $\bfUpsilon$ to be the analysis map of an orthonormal basis for $\calE$ yields an equivalent sequence $(\widehat{\calU}_i)_{i=1}^{n}$ of subspaces of $\bbF^d$.

Angles between subspaces can be understood as follows.
For any distinct $i,j\in[n]$,
the singular values of $\bfPhi_i^*\bfPhi_j$ are bounded above by $\norm{\bfPhi_i^*\bfPhi_j}_{\op}
\leq\norm{\bfPhi_i^*}_{\op}\norm{\bfPhi_j^{}}_{\op}=1$,
and so are of the form $(\cos(\theta_{i,j;k}))_{k=1}^{r}$ for some unique nondecreasing sequence of \textit{principal angles} $(\theta_{i,j;k})_{k=1}^{r}$ in $[0,\frac\pi 2]$.
By Lemma~\ref{lem.equivalence},
these angles are preserved by equivalence,
and in particular are invariant with respect to one's choice of isometries for $(\calU_i)_{i=1}^{n}$.
The \textit{spectral distance}~\cite{DhillonHST08} between two of these subspaces is
\begin{equation}
\label{eq.spectral distance}
\dist_\rms(\calU_i,\calU_j)
:=(1-\norm{\bfPhi_i^*\bfPhi_j^{}}_{\op}^2)^{\frac12}
=\min_{k\in[r]}\sin(\theta_{i,j;k}).
\end{equation}
With respect to it, an optimal code in the corresponding Grassmannian equates to $n$ subspaces of $\calE$, each of dimension $r$, for which the smallest principal angle between any pair of them is as large as possible.
In light of~\eqref{eq.spectral distance},
this equates to minimizing block coherence~\eqref{eq.block coherence}.

Later on, we use the following straightforward bounds:
\begin{align}
\nonumber
\tfrac1{n(n-1)}\sum_{i=1}^{n}\sum_{j\neq i}\tfrac1r\norm{\bfPhi_i^*\bfPhi_j}_\Fro^2
&\leq\tfrac1{n(n-1)}\sum_{i=1}^{n}\sum_{j\neq i}\norm{\bfPhi_i^*\bfPhi_j}_\op^2\\
\label{eq.equi-isoclinic bound}
&\leq\max_{i\neq j}\norm{\bfPhi_i^*\bfPhi_j}_\op^2.
\end{align}
Clearly, equality holds throughout~\eqref{eq.equi-isoclinic bound} if and only if $\theta_{i,j;k}$ is constant over all distinct $i,j\in[n]$ and $k\in[r]$.
When this occurs, $(\calU_i)_{i=1}^n$ is said to be \textit{equi-isoclinic}~\cite{LemmensS73b}.
This equates to the existence of some $\sigma\geq0$ such that
\begin{equation}
\label{eq.equi-isoclinic}
\bfPhi_i^*\bfPhi_j^{}\bfPhi_j^{*}\bfPhi_i^{}
=\sigma^2\brI,
\ \forall\, i,j\in[n],\, i\neq j.
\end{equation}
For any $\sigma\geq0$,
\eqref{eq.equi-isoclinic} equates,
via the appropriate multiplication by $\bfPhi_i^{}$ and $\bfPhi_i^*$, to having
\begin{equation}
\label{eq.equi-isoclinic projections}
\bfPi_i\bfPi_j\bfPi_i
=\sigma^2\bfPi_i,
\ \forall\, i,j\in[n],\, i\neq j.
\end{equation}
The real instance of the following necessary condition on the existence of equi-isoclinic subspaces was given in~\cite{LemmensS73b}.
In the appendix, we verify that its proof directly generalizes to the complex setting.
Results of this type are called \textit{Gerzon bounds}.

\begin{lemma}[\cite{LemmensS73b}]
\label{lem.Gerzon}
If a space of dimension $d$ contains $n$ nonidentical equi-isoclinic subspaces of dimension $r$ then
\begin{equation}
\label{eq.Gerzon bound}
n\leq\left\{\begin{array}{ll}
\tfrac12 d(d+1)-\tfrac12 r(r+1)+1,&\ \bbF=\bbR,\smallskip\\
d^2-r^2+1,&\ \bbF=\bbC.
\end{array}\right.
\end{equation}
\end{lemma}

In~\cite{LemmensS73b}, the authors state that they ``have reasons to believe that this bound is too large" in general.
These reasons presumably include Theorem~5.1 of~\cite{LemmensS73b},
which shows that this upper bound is particularly poor when $d=2r$ and $\bbF=\bbR$.
For context, when $r=1$ and $\bbF=\bbC$, it is famously conjectured~\cite{Zauner99,FuchsHS17} that equality is achieved in~\eqref{eq.Gerzon bound} for every $d\geq2$.
When $r=1$ and $\bbF=\bbR$, equality is achieved in~\eqref{eq.Gerzon bound} when $d\in\set{2,3,7,23}$,
and it is conjectured~\cite{GodsilR01,Gillespie18} that these are the only $d$ for which this occurs.
Our Theorem~\ref{thm.RH EITFF existence} contributes to this discussion,
showing that~\eqref{eq.Gerzon bound} is a weak bound on the existence of EITFFs with $d=2r$ and $\bbF=\bbC$.

\subsection{Equi-isoclinic tight fusion frames}

For $(\calU_i)_{i=1}^{n}$ and $(\bfPhi_i)_{i=1}^{n}$ defined as above, let $\bfPhi$ be the synthesis map of the concatenation $(\bfphi_{i,k})_{(i,k)\in[n]\times[r]}$
of the chosen orthonormal bases.
The associated \textit{fusion frame operator} is
\begin{equation}
\label{eq.fusion frame operator}
\bfPhi\bfPhi^*
=\sum_{i=1}^n\sum_{k=1}^{r}\bfphi_{i,k}^{}\bfphi_{i,k}^*
=\sum_{i=1}^{n}\bfPhi_i^{}\bfPhi_{i}^{*}
=\sum_{i=1}^{n}\bfPi_i.
\end{equation}
Meanwhile, the associated \textit{fusion Gram matrix} $\bfPhi^*\bfPhi$ satisfies
\begin{equation*}
(\bfPhi^*\bfPhi)\bigl((i,k),(j,l)\bigr)
=\ip{\bfphi_{i,k}}{\bfphi_{j,l}}
=(\bfPhi_i^*\bfPhi_j^{})(k,l),
\end{equation*}
for any $i,j\in[n]$, $k,l\in[r]$,
and so is a block matrix with the cross-Gram matrix $\bfPhi_i^*\bfPhi_j^{}\in\bbF^{r\times r}$ as its $(i,j)$th submatrix.
In particular, each of its diagonal blocks is $\brI$.

We call $(\calU_i)_{i=1}^{n}$ a \textit{tight fusion frame} for $\calE$ if its fusion frame operator~\eqref{eq.fusion frame operator} is a positive multiple of the identity on $\calE$,
namely if $\bfPhi\bfPhi^*=\sum_{i=1}^{n}\bfPi_i=a\brI$ for some $a>0$.
By Lemma~\ref{lem.equivalence}, this property is preserved under equivalence.
Taking the trace and rank of this equation reveals that this can only occur if  $a=\tfrac{nr}{d}\geq 1$, where $d$ is the dimension of $\calE$.

If $(\calU_i)_{i=1}^{n}$ is a tight fusion frame for $\calE$,
squaring and taking the trace of $\frac{d}{nr}\bfPhi^*\bfPhi$ reveals that it is a rank-$d$ projection,
and so $\brI-\frac{d}{nr}\bfPhi^*\bfPhi$ is a projection of rank $nr-d$.
Provided $nr>d$,
there thus exists a sequence $(\widetilde{\bfphi}_{i,k})_{(i,k)\in[n]\times[r]}$ of vectors in a space of dimension $nr-d$, unique up to unitaries, whose synthesis map \smash{$\widetilde{\bfPhi}$} satisfies $\widetilde{\bfPhi}^*\widetilde{\bfPhi}=\tfrac{nr}{nr-d}(\brI-\frac{d}{nr}\bfPhi^*\bfPhi)$.
Here, since $\bfPhi^*\bfPhi$ is a block matrix, each of whose diagonal blocks is $\brI\in\bbF^{r\times r}$,
the same is true for \smash{$\widetilde{\bfPhi}^*\widetilde{\bfPhi}$}.
As such, each synthesis map $\widetilde{\bfPhi}_i$ of $(\widetilde{\bfphi}_{i,k})_{k=1}^r$ is an isometry,
and together they satisfy
\begin{equation}
\label{eq.Naimark}
\widetilde{\bfPhi}_i^*\widetilde{\bfPhi}_j^{}=-\tfrac{d}{nr-d}\bfPhi_i^*\bfPhi_j^{},
\ \forall\, i,j\in[n],\, i\neq j.
\end{equation}
Moreover, since \smash{$\widetilde{\bfPhi}^*\widetilde{\bfPhi}$} has eigenvalue $\tfrac{nr}{nr-d}$ with multiplicity $nr-d$,
taking the singular value decomposition of \smash{$\widetilde{\bfPhi}$} reveals that \smash{$\widetilde{\bfPhi}\widetilde{\bfPhi}^*=\tfrac{nr}{nr-d}\brI$}.
As such, the $r$-dimensional images of $(\widetilde{\bfPhi}_i)_{i=1}^{n}$ form a tight fusion frame for a space of dimension $nr-d$.
It is called a \textit{Naimark complement} of $(\calU_i)_{i=1}^{n}$.
In light of~\eqref{eq.Naimark} and Lemma~\ref{lem.equivalence},
it is unique up to equivalence.

Regardless of whether $(\calU_i)_{i=1}^{n}$ is a tight fusion frame for $\calE$,
a direct calculation gives
\begin{align}
\label{eq.frame potential bound derivation}
0&\leq\norm{\tfrac{nr}{d}\brI-\bfPhi\bfPhi^*}_{\Fro}^2\\
\nonumber&\textstyle=\Tr\bigl[(\tfrac{nr}{d}\brI-\sum_{i=1}^{n}\bfPhi_i^{}\bfPhi_i^*)^2\bigr]\\
&\nonumber=-\tfrac{nr(nr-d)}{d}
+\sum_{i=1}^{n}\sum_{j\neq i}\norm{\bfPhi_i^*\bfPhi_j}_\Fro^2.
\end{align}
Rearranging this gives the following equivalent inequality:
\begin{equation}
\label{eq.frame potential bound}
\tfrac{nr-d}{d(n-1)}
\leq\tfrac1{n(n-1)}
\sum_{i=1}^{n}\sum_{j\neq i}\tfrac1r\norm{\bfPhi_i^*\bfPhi_j}_\Fro^2.
\end{equation}
Notably, equality holds in~\eqref{eq.frame potential bound} if and only if it holds in~\eqref{eq.frame potential bound derivation} namely if and only if $(\calU_i)_{i=1}^{n}$ is a tight fusion frame for $\calE$.

Combining~\eqref{eq.frame potential bound} and~\eqref{eq.equi-isoclinic bound} gives the Welch bound~\eqref{eq.Welch bound}.
Moreover, this argument implies that the following are equivalent:
(i) $(\bfPhi_i)_{i=1}^{n}$ achieves equality in~\eqref{eq.Welch bound};
(ii) $(\bfPhi_i)_{i=1}^{n}$ achieves equality in both~\eqref{eq.frame potential bound} and~\eqref{eq.equi-isoclinic bound};
(iii) $(\calU_i)_{i=1}^{n}$ is a tight fusion frame for $\calE$ that happens to be equi-isoclinic, that is, is an EITFF for $\calE$;
(iv) $(\bfPhi_i)_{i=1}^{n}$ satisfies~\eqref{eq.equi-isoclinic} with
\smash{$\sigma^2=\frac{nr-d}{d(n-1)}$}.

By Lemma~\ref{lem.equivalence}, tightness and equi-isoclinism are preserved by equivalence.
As such, the existence of an $\EITFF_\bbF(d,r,n)$,
namely an EITFF for a space over $\bbF$ of dimension $d$ that consists of $n$ of its subspaces of dimension $r$, depends only on $\bbF$, $d\geq r\geq1$ and $n\geq \frac dr$.
Moreover,
if so desired,
we can let $\calE=\bbF^d$ and let each $\bfPhi_i$ be a $d\times r$ matrix with orthonormal columns.
By~\eqref{eq.Welch bound},
each such $\EITFF_\bbF(d,r,n)$ is an $n$-element optimal code in the Grassmannian of all subspaces of $\bbF^d$ of dimension $r$, with respect to the spectral distance~\eqref{eq.spectral distance}.

By~\eqref{eq.Naimark},
a Naimark complement of any $\EITFF_\bbF(d,r,n)$ with $nr>d$ is an $\EITFF_\bbF(nr-d,r,n)$.
Any $\EITFF_\bbF(d,r,n)$ of one of the following three types is said to be \textit{trivial}:
its subspaces are mutually orthogonal, i.e., $d=nr$;
it consists of multiple copies of the entire space, i.e., $d=r$;
Naimark complements of the latter, i.e., $d=(n-1)r$.
Every EITFF for $\calE$ is also an \textit{equichordal} tight fusion frame,
and so is also an optimal Grassmannian code with respect to Conway, Hardin and Sloane's \textit{chordal distance}~\cite{ConwayHS96}.

\subsection{Radon--Hurwitz theory}

Our main results connect $\EITFF_\bbF(d,r,n)$ having $d=2r$ with solutions of the classical Radon--Hurwitz equations, which we now review.
We call a matrix $\bfA\in\bbF^{r\times r}$ a \textit{scaled unitary} if it is a scalar multiple of a unitary.
We claim that this occurs if and only if $\bfA^*\bfA=z\brI$ for some $z\in\bbF$.
Clearly the former implies the latter.
Conversely, taking traces gives $z=\tfrac1r\Tr(\bfA^*\bfA)\geq0$.
If $z>0$ then $z^{-\frac12}\bfA$ is unitary.
When instead $z=0$, $\bfA$ is necessarily $\bfzero$ since $\ker(\bfA^*\bfA)=\ker(\bfA)$ in general.

Regardless of whether $\bbF$ is $\bbR$ or $\bbC$,
we here regard $\bbF^{r\times r}$ as a vector space over $\bbR$.
The \textit{Radon--Hurwitz number} $\rho_\bbF(r)$ is the maximum dimension of a subspace of $\bbF^{r\times r}$ whose every member is a scaled unitary.
For example, letting
\begin{equation}
\label{eq.def of T,M,R}
\brI:=\left[\begin{smallmatrix}1&0\\0&1\end{smallmatrix}\right],\
\brM:=\left[\begin{smallmatrix}1&\phantom{-}0\\0&-1\end{smallmatrix}\right],\
\brT:=\left[\begin{smallmatrix}0&1\\1&0\end{smallmatrix}\right],\
\brR:=\left[\begin{smallmatrix}0&-1\\1&\phantom{-}0\end{smallmatrix}\right],
\end{equation}
note that $a\brI+b\brR=\left[\begin{smallmatrix}a&-b\\b&\hfill a\end{smallmatrix}\right]$ is a real scaled unitary for any $a,b\in\bbR$,
and so $\rho_\bbR(2)\geq 2$.
Similarly, $\rho_\bbC(2)\geq 4$ since
\begin{equation*}
a\brI+b\rmi\brM+c\brT+d\rmi\brR
=\left[\begin{smallmatrix}
a+\rmi b&-c+\rmi d\\
c+\rmi d&\phantom{-}a-\rmi b
\end{smallmatrix}\right]
\end{equation*}
is a complex scaled unitary for any $a,b,c,d\in\bbR$ where $\brI$, $\rmi\brM$, $\brT$ and $\rmi\brR$ are linearly independent.
A series of classical results due to Radon~\cite{Radon22}, Hurwitz~\cite{Hurwitz23} and others~\cite{Adams62,AdamsLP65} yield the nonobvious yet simply stated formula for $\rho_\bbF(r)$ given in~\eqref{eq.RH number};
see also~\cite{AdamsLP66,Au-Yeung71,Au-YeungC06} for related work.
In this subsection,
we review ideas from this literature that pertains to our work here.

For brevity's sake, we often abbreviate ``Radon--Hurwitz" as ``$\rho$" when using it as a prefix.
For instance, we call a subspace $\calS$ of $\bbF^{r\times r}$ a \textit{$\rho$-space} if every $\bfA\in\calS$ is a scaled unitary.
Thus, $\rho_\bbF(r)$ is the maximum dimension of a $\rho$-space in $\bbF^{r\times r}$.
If a $\rho$-space contains $\bfA$ and $\bfB$ then $\bfA+\bfB$ is a scaled unitary and $(\bfA+\bfB)^*(\bfA+\bfB)=\bfA^*\bfA+\bfA^*\bfB+\bfB^*\bfA+\bfB^*\bfB$, implying there exists $z\in\bbF$ such that
\begin{equation}
\label{eq.RH property}
\bfA^*\bfB^{}+\bfB^*\bfA^{}
=2z\brI.
\end{equation}
Taking the trace of~\eqref{eq.RH property} gives that this $z$ is necessarily
\begin{equation}
\label{eq.RH inner product}
\ip{\bfA}{\bfB}_\rho:=\tfrac1r\real(\Tr(\bfA^*\bfB))\in\bbR.
\end{equation}
Clearly, \eqref{eq.RH inner product} defines an inner product on a $\rho$-space in $\bbF^{r\times r}$.
More generally, we call a subset $\calS$ of $\bbF^{r\times r}$ a \textit{$\rho$-set} if for any $\bfA,\bfB\in\calS$ there exists $z\in\bbF$ such that~\eqref{eq.RH property} holds.
For any $\rho$-set $\calS$,
taking $\bfA=\bfB$ in~\eqref{eq.RH property} gives that any $\bfA\in\calS$ is a scaled unitary,
whereupon~\eqref{eq.RH property} in general gives that $\bfA+\bfB$ is a scaled unitary for any $\bfA,\bfB\in\calS$.
It follows that the real span of any $\rho$-set is a $\rho$-space.
For any $\bfA$ and $\bfB\neq\bfzero$ in a $\rho$-set,
multiplying~\eqref{eq.RH property} on the left and right by $\bfB$ and $\bfB^*$, respectively, gives $\bfA^{}\bfB^{*}+\bfB^{}\bfA^{*}=2\ip{\bfA}{\bfB}_\rho\brI$;
clearly the latter also holds when $\bfB=\bfzero$.
As such, if $\calS$ is a $\rho$-set then so is $\calS^*:=\set{\bfA^*: \bfA\in\calS}$ and moreover,
for any $\bfA,\bfB\in\calS$,
\begin{equation}
\label{eq.RH set relations}
\bfA^*\bfB^{}+\bfB^*\bfA^{}
=2\ip{\bfA}{\bfB}_\rho\,\brI
=\bfA^{}\bfB^{*}+\bfB^{}\bfA^{*}.
\end{equation}

Any $\rho$-space is thus a finite-dimensional real Hilbert space under~\eqref{eq.RH inner product} whose unit norm elements are unitaries.
We say that a finite sequence $(\bfC_i)_{i=1}^{m}$ in $\bbF^{r\times r}$ is \textit{$\rho$-orthonormal} if it is orthonormal with respect to~\eqref{eq.RH inner product}, that is, if each $\bfC_i$ is unitary and they satisfy the classical \textit{Radon--Hurwitz equations}
\begin{equation}
\label{eq.RH equations}
\bfC_i^*\bfC_j^{}+\bfC_j^*\bfC_i^{}
=\bfzero,
\ \forall\,i,j\in[m],\,i\neq j.
\end{equation}
For example,
from~\eqref{eq.def of T,M,R},
it follows that
$(\brI,\brR)$ is $\rho$-orthonormal in $\bbR^{2\times 2}$ and $(\brI,\brR,\rmi\brM,\rmi\brT)$ is $\rho$-orthonormal in $\bbC^{2\times 2}$.
Since the members of any such sequence form a $\rho$-set,
their real span is a $\rho$-space.
It follows that $\rho_\bbF(r)$ is the largest positive integer $m$ for which there exist $m$ unitary matrices in $\bbF^{r\times r}$ that satisfy~\eqref{eq.RH equations},
and moreover,
that any $\rho$-space of a given dimension contains a $\rho$-space of any lesser dimension.

We say that a sequence in $\bbF^{r\times r}$ is \textit{equivalent} to another such sequence $(\bfA_i)_{i=1}^{m}$ if it is of the form $(\bfU^*\bfA_i\bfV)_{i=1}^{m}$ for some unitaries $\bfU$ and $\bfV$ in $\bbF^{r\times r}$.
By~\eqref{eq.RH inner product}, this notion of equivalence preserves \mbox{$\rho$-orthonormality}.
In particular, unitaries $(\bfC_i)_{i=1}^{m}$ in $\bbF^{r\times r}$ are $\rho$-orthonormal if and only if $(\bfC_m^*\bfC_i^{})_{i=1}^{m}$ are as well.
Moreover,
$(\brI,\bfC)$ in $\bbF^{r\times r}$ is $\rho$-orthonormal if and only if $\bfC$ is both \textit{skew-Hermitian}, namely $\bfC^*=-\bfC$, and unitary.
It follows that $\bbF^{r\times r}$ contains $m$ $\rho$-orthonormal matrices if and only if it contains $m-1$ skew-Hermitian unitaries, any two of which anticommute.
Exploiting this idea yields the following characterization of complex $\rho$-orthonormality.
Its proof, given in the appendix, informs our work in Section~IV.
It implies that $\rho_\bbC(r)=\rho_\bbC(\frac r2)+2$ when $r$ is even and $\rho_\bbC(r)=2$ when $r$ is odd,
yielding the formula for $\rho_\bbC(r)$ given in~\eqref{eq.RH number}.

\begin{lemma}[Folklore]
\label{lem.complex RH}
If $\bbC^{r\times r}$ contains a $\rho$-orthonormal sequence of $m\geq 3$ matrices then $2\mid r$ and it equates to some $(\bfC_i)_{i=1}^{m}$ with
$\bfC_m=\left[\begin{smallmatrix}
\brI&\bfzero\\
\bfzero&\brI
\end{smallmatrix}\right]$,
$\bfC_{m-1}
=\rmi\brM\otimes\brI
=\rmi\left[\begin{smallmatrix}
\brI&\phantom{-}\bfzero\\
\bfzero&-\brI
\end{smallmatrix}\right]$,
and
\begin{equation}
\label{eq.complex reduction}
\bfC_i
=\left[\begin{smallmatrix}
\bfzero&-\widehat{\bfC}_i^*\\
\widehat{\bfC}_i^{}&\bfzero
\end{smallmatrix}\right],
\ \forall\,i\in[m-2],
\end{equation}
where $(\widehat{\bfC}_i)_{i=1}^{m-2}$ is some $\rho$-orthonormal sequence in $\bbC^{\frac{r}{2}\times\frac{r}{2}}$.
Conversely,
any such $(\bfC_i)_{i=1}^{m}$ is $\rho$-orthonormal.
\end{lemma}

Lemma~\ref{lem.complex RH} transforms a maximal $\rho$-orthonormal sequence in $\bbC^{\frac{r}{2}\times\frac{r}{2}}$ into one in $\bbC^{r\times r}$.
For example,
it transforms $(\rmi,1)$ in $\bbC^{1\times1}$ into $(\rmi\brT,\brR,\rmi\brM,\brI)$ in $\bbC^{2\times2}$,
and transforms the latter into a maximal $\rho$-orthonormal sequence
in $\bbC^{4\times4}$, namely
\mbox{$(\rmi\brT\otimes\brT,\
\brT\otimes\brR,\
\rmi\brT\otimes\brM,\
\brR\otimes\brI,\
\rmi\brM\otimes\brI,\
\brI\otimes\brI)$}.

When $\bbF=\bbR$,
a more delicate approach is needed.
Comparing the values $\rho_\bbR(r)=8b+2^c$ and $\rho_\bbC(r)=8b+2c+2$  from~\eqref{eq.RH number} reveals that $\rho_\bbR(r)=\rho_\bbC(r)$ when $c=3$,
namely when $\frac{r}{8}$ is an odd multiple of a power of $16$,
and that $\rho_\bbR(r)<\rho_\bbC(r)$ when instead $c\in\set{0,1,2}$.
Since every real matrix is a complex one, $\rho_\bbC(r)$ is an upper bound on the maximal number of \mbox{$\rho$-orthonormal} matrices in $\bbR^{r\times r}$.
When $c\in\set{0,1,2}$, we know of no short and elementary proof of the classical fact that $\rho_\bbR(r)=8b+2^c$ is an upper bound on this same number;
one nontrivial proof of it~\cite{Rajwade93} involves a series of four various nontrivial generalizations of~\eqref{eq.complex reduction} that convert $\rho$-orthonormal matrices in $\bbR^{r\times r}$ into eight fewer such matrices in $\bbR^{\frac{r}{16}\times\frac{r}{16}}$.

In contrast, it is relatively straightforward to explicitly construct a maximal number of unitaries in $\bbR^{r\times r}$ that satisfy~\eqref{eq.RH equations}.
Taking the first of these to be $\brI$,
it suffices to have $m:=8b+2^c-1$ anticommuting skew-Hermitian unitaries in $\bbR^{r\times r}$.
For example, recalling~\eqref{eq.def of T,M,R},
consider the matrix $\brR$ when $r=2$ and $m=1$,
the matrices
$(\brI\otimes\brR$, $\brR\otimes\brT$, $\brR\otimes\brM)$ when $r=4$ and $m=3$, the matrices
($\brM\otimes\brM\otimes\brR$,
$\brM\otimes\brT\otimes\brR$,
$\brM\otimes\brR\otimes\brI$,
$\brT\otimes\brR\otimes\brM$,
$\brT\otimes\brR\otimes\brT$,
$\brT\otimes\brI\otimes\brR$,
$\brR\otimes\brI\otimes\brI$)
when $r=8$ and $m=7$, and the matrices
\begin{equation}
\label{eq.RH matrices of size 16}
\begin{array}{ccccccccc}
\brR&\otimes&\brT&\otimes&\brT&\otimes&\brT,\\
\brT&\otimes&\brR&\otimes&\brT&\otimes&\brM,\\
\brT&\otimes&\brM&\otimes&\brR&\otimes&\brT,\\
\brT&\otimes&\brT&\otimes&\brM&\otimes&\brR,
\end{array}
\qquad
\begin{array}{ccccccc}
\brR&\otimes&\brM&\otimes&\brM&\otimes&\brM,\\
\brM&\otimes&\brR&\otimes&\brM&\otimes&\brT,\\
\brM&\otimes&\brT&\otimes&\brR&\otimes&\brM,\\
\brM&\otimes&\brM&\otimes&\brT&\otimes&\brR,
\end{array}
\end{equation}
when $r=16$ and $m=8$.
Moreover,
as noted in~\cite{Au-YeungC06},
combining the above matrices $(\bfE_i)_{i=1}^{8}$ in $\bbR^{16\times 16}$ with any anticommuting skew-Hermitian matrices $(\bfC_i)_{i=1}^{m}$ in $\bbR^{r\times r}$ as follows yields $m+8$ such matrices in $\bbR^{16r\times 16r}$:
\begin{equation}
\label{eq.real inflation}
(\brR\otimes\brR\otimes\brR\otimes\brR\otimes\bfC_i)_{i=1}^m\sqcup
(\bfE_i\otimes\brI)_{i=1}^{8}.
\end{equation}

\subsection{Regular simplices}

Our characterization of $\EITFF_\bbF(d,r,n)$ with $d=2r$ involves a deep connection with $\rho$-spaces, but we will not encounter the latter directly as the span of $\rho$-orthonormal sequences.
Instead, we will see the following kind of object.
We say that a sequence $(\bfphi_i)_{i=1}^{m}$ of $m\geq 2$ unit norm vectors in a real space is a \textit{(regular) simplex} if
$\ip{\bfphi_i}{\bfphi_j}=-\tfrac1{m-1}$ for any distinct $i,j\in[m]$.
That is, in this paper, a simplex is a sequence of vectors in a real space whose Gram matrix is $\tfrac1{m-1}(m\brI-\brJ)$ where $\brJ$ is an $m\times m$ all-ones matrix.
Such a matrix has rank $m-1$ and is positive semidefinite.
As such, a real space $\calE$ contains an \mbox{$m$-vector} simplex if and only if
$m\leq\dim(\calE)+1$.
To recursively construct a matrix $\bfPsi_m$ whose $m$ columns form a simplex in $\bbR^{m-1}$,
let $\bfPsi_2:=\left[\begin{array}{cc}1&-1\end{array}\right]$ and
\begin{equation}
\label{eq.simplex recursion}
\bfPsi_m
:=\frac1{m-1}\left[\begin{array}{cc}
m-1&-\bfone^{\top}\medskip\\
\bfzero&\sqrt{m(m-2)}\,\bfPsi_{m-1}
\end{array}\right],
\ m\geq 3.
\end{equation}
In the next result, we summarize some elementary but obscure facts about simplices that we will use later on:

\begin{lemma}[Folklore]
\label{lem.simplex}
A sequence $(\bfphi_i)_{i=1}^{m}$ of $m\geq 2$ vectors in a real space is a simplex if and only if there exists an orthonormal basis $(\bfupsilon_i)_{i=1}^{m-1}$ for $\Span(\bfphi_i)_{i=1}^{m}$ such that
\begin{equation}
\label{eq.simplex basis}
\bfphi_j=\sum_{i=1}^{m-1}\bfPsi_m(i,j)\bfupsilon_i,
\ \forall\, j\in[m],
\end{equation}
where $\bfPsi_m$ is given by~\eqref{eq.simplex recursion}.
In that case, the following hold:
\begin{enumerate}
\renewcommand{\labelenumi}{(\alph{enumi})}
\item
$\bfupsilon_1=\bfphi_1$.
\item
$\bfupsilon_j\in\Span(\bfphi_i)_{i=1}^{j}$ for each $j\in[m-1]$.
\item
$\bfupsilon_{m-1}$ is a positive multiple of $\bfphi_{m-1}-\bfphi_m$.
\end{enumerate}
\end{lemma}

For a proof, see the appendix.
The interested reader can verify that for a given simplex $(\bfphi_i)_{i=1}^{m}$,
the associated orthonormal basis $(\bfupsilon_i)_{i=1}^{m-1}$ of Lemma~\ref{lem.simplex} is unique, and can be iteratively constructed by orthogonalizing  $(\bfphi_i)_{i=1}^{m-1}$ via the Gram--Schmidt process.

\renewcommand{\thesection}{\Roman{section}}
\section{Characterizing Radon--Hurwitz EITFFs}
\renewcommand{\thesection}{\arabic{section}}

In this section,
we characterize $\EITFF_\bbF(d,r,n)$ for which $d=2r$.
When $n\in\set{2,3}$, such EITFFs are trivial.
When $n\geq 3$, they equate to $\EITFF_\bbF((n-2)r,r,n)$ via Naimark complementation.
We begin with the following result, which characterizes the isometries of such EITFFs.

\begin{theorem}
\label{thm.RH EITFF isometries}
For any integers $r\geq 1$ and $n\geq 3$,
any $\EITFF_\bbF(2r,r,n)$ is equivalent to an EITFF for $\bbF^{2r}$ with isometries $(\bfPhi_i)_{i=1}^n$ of the form
\begin{equation}
\label{eq.RH EITFF isometries}
\bfPhi_i
=\left[\begin{array}{c}
\alpha\brI\\
\beta\bfB_i
\end{array}\right],
\ \forall\,i\in[n-1],
\quad
\bfPhi_n
=\left[\begin{array}{c}
\brI\\
\bfzero
\end{array}\right],
\end{equation}
where $\brI$ is the identity matrix in $\bbF^{r\times r}$,
$\alpha$ and $\beta$ are the scalars
\begin{equation}
\label{eq.RH EITFF parameters}
\alpha=\alpha(n):=\bigparen{\tfrac{n-2}{2n-2}}^{\frac12},
\quad
\beta=\beta(n):=\bigparen{\tfrac{n}{2n-2}}^{\frac12},
\end{equation}
and $(\bfB_i)_{i=1}^{n-1}$ are unitaries in $\bbF^{r\times r}$ that satisfy
\begin{equation}
\label{eq.RH simplex}
\bfB_i^*\bfB_j^{}+\bfB_j^*\bfB_i^{}=-\tfrac{2}{n-2}\brI,
\ \forall\,i,j\in[n-1],\,i\neq j.
\end{equation}
Conversely, for any such $(\bfB_i)_{i=1}^{n-1}$,
defining $(\bfPhi_i)_{i=1}^{n}$ by~\eqref{eq.RH EITFF isometries} and~\eqref{eq.RH EITFF parameters} yields the isometries of an $\EITFF_\bbF(2r,r,n)$ for $\bbF^{2r}$.
\end{theorem}

\begin{proof}
Let $(\widehat{\bfPhi}_i)_{i=1}^n$ be isometries of some $\EITFF_\bbF(2r,r,n)$ $(\widehat{\calU}_i)_{i=1}^{n}$ for a space $\widehat{\calE}$.
Define $\bfPhi_n\in\bbF^{2r\times r}$ as in~\eqref{eq.RH EITFF isometries}.
Since $\bfPhi_n^*\bfPhi_n^{}=\brI=\widehat{\bfPhi}_n^*\widehat{\bfPhi}_n^{}$,
there is a unitary $\bfUpsilon_0$ from the image of $\bfPhi_n$ to that of $\widehat{\bfPhi}_n$ such that $\widehat{\bfPhi}_n=\bfUpsilon_0\bfPhi_n$.
Since $\widehat{\calE}$ has dimension $2r$ over $\bbF$,
$\bfUpsilon_0$ extends to a unitary $\bfUpsilon:\bbF^{2r}\rightarrow\widehat{\calE}$ such that $\widehat{\bfPhi}_n=\bfUpsilon\bfPhi_n$.
Letting $\bfOmega_i:=\bfUpsilon^*\widehat{\bfPhi}_i\in\bbF^{2r\times r}$ for all $i\in[n]$ thus yields isometries $(\bfOmega_i)_{i=1}^{n}$ of an (equivalent) EITFF for $\bbF^{2r}$ for which $\bfOmega_n^{}=\bfUpsilon^*\widehat{\bfPhi}_n=\bfPhi_n$.
As detailed in the previous section, this implies that
$\bfOmega_j^*\bfOmega_i^{}\bfOmega_i^*\bfOmega_j^{}=\sigma^2\brI$ for all distinct $i,j\in[n]$ where necessarily
\begin{equation}
\label{eq.pf of RH EITFF existence 1}
\sigma
=\bigbracket{\tfrac{nr-d}{d(n-1)}}^{\frac12}
=\bigbracket{\tfrac{nr-2r}{2r(n-1)}}^{\frac12}
=\bigparen{\tfrac{n-2}{2n-2}}^{\frac12}
=\alpha
>0.
\end{equation}
For each $i\in[n]$,
take $\bfOmega_{i,1},\bfOmega_{i,2}\in\bbF^{r\times r}$ such that
\begin{equation*}
\bfOmega_i
=\left[\begin{array}{c}
\bfOmega_{i,1}\\
\bfOmega_{i,2}
\end{array}\right],
\text{ and so }
\bfOmega_n^*\bfOmega_i^{}
=\left[\begin{array}{cc}\brI&\bfzero\end{array}\right]
\left[\begin{array}{c}
\bfOmega_{i,1}\\
\bfOmega_{i,2}
\end{array}\right]
=\bfOmega_{i,1}.
\end{equation*}
For any $i\in[n-1]$,
$\bfZ_i:=\alpha^{-1}\bfOmega_{i,1}$ is unitary since
\begin{equation*}
\bfZ_i^*\bfZ_i^{}
=\tfrac1{\alpha^2}\bfOmega_{i,1}^*\bfOmega_{i,1}^{}
=\tfrac1{\sigma^2}\bfOmega_i^*\bfOmega_n^{}\bfOmega_n^*\bfOmega_i^{}
=\brI.
\end{equation*}
Letting
$\bfPhi_i^{}:=\bfOmega_i^{}\bfZ_i^*$ for all $i\in[n-1]$ and $\bfZ_n:=\brI$,
it follows that
$\widehat{\bfPhi}_i
=\bfUpsilon\bfOmega_i
=\bfUpsilon\bfPhi_i\bfZ_i$ for all $i\in[n]$,
and so $(\bfPhi_i)_{i=1}^{n}$ is a sequence of isometries that is equivalent to $(\widehat{\bfPhi}_i)_{i=1}^{n}$.
Next,
for any $i\in[n-1]$,
let $\bfB_i:=\beta^{-1}\bfOmega_{i,2}^{}\bfZ_i^*\in\bbF^{r\times r}$ and note
\begin{equation*}
\bfPhi_i
=\bfOmega_i^{}\bfZ_i^*
=\left[\begin{array}{c}
\bfOmega_{i,1}\\
\bfOmega_{i,2}
\end{array}\right]\bfZ_i^*
=\left[\begin{array}{r}
\alpha\bfZ_i^{}\\
\beta\bfB_i\bfZ_i
\end{array}\right]\bfZ_i^*
=\left[\begin{array}{c}
\alpha\brI\\
\beta\bfB_i
\end{array}\right].
\end{equation*}
Thus, for any $i,j\in[n-1]$,
\begin{equation}
\label{eq.pf of RH EITFF existence 2}
\bfPhi_i^*\bfPhi_j^{}
=\left[\begin{array}{cc}\alpha\brI&\beta\bfB_i^*\end{array}\right]
\left[\begin{array}{c}\alpha\brI\\\beta\bfB_j\end{array}\right]
=\alpha^2\brI+\beta^2\bfB_i^*\bfB_j^{}.
\end{equation}
Since $\alpha^2+\beta^2=1$,
taking $i=j$ in~\eqref{eq.pf of RH EITFF existence 2} gives that each $\bfB_i$ is unitary.
Moreover, since $(\bfPhi_i)_{i=1}^{n}$ is a sequence of isometries of an $\EITFF_{\bbF}(2r,r,n)$,
\eqref{eq.pf of RH EITFF existence 1} and~\eqref{eq.pf of RH EITFF existence 2} imply
\begin{align*}
\alpha^2\brI
&=\bfPhi_j^*\bfPhi_i^{}\bfPhi_i^*\bfPhi_j^{}\\
&=(\alpha^2\brI+\beta^2\bfB_j^*\bfB_i^{})(\alpha^2\brI+\beta^2\bfB_i^*\bfB_j^{})\\
&=(\alpha^4+\beta^4)\brI+\alpha^2\beta^2(\bfB_i^*\bfB_j^{}+\bfB_j^*\bfB_i^{})
\end{align*}
for any distinct $i,j\in[n-1]$.
Since $\alpha^2-\alpha^4-\beta^4=-\tfrac{2}{n-2}\alpha^2\beta^2$ by~\eqref{eq.RH EITFF parameters},
it follows that $(\bfB_i)_{i=1}^{n}$ satisfies~\eqref{eq.RH simplex}.

Conversely, if $(\bfB_i)_{i=1}^{n-1}$ is any sequence of unitaries in $\bbF^{r\times r}$ that satisfies~\eqref{eq.RH simplex},
a straightforward reversal of these arguments gives that~\eqref{eq.RH EITFF isometries}
and~\eqref{eq.RH EITFF parameters} define isometries $(\bfPhi_i)_{i=1}^{n}$ in $\bbF^{2r\times r}$ that satisfy $\alpha^2\brI=\bfPhi_j^*\bfPhi_i^{}\bfPhi_i^*\bfPhi_j^{}$ for any distinct $i,j\in[n]$,
and so whose images form an $\EITFF_\bbF(2r,r,n)$ for $\bbF^{2r}$.
\end{proof}

By Theorem~\ref{thm.RH EITFF isometries},
the existence of an $\EITFF_\bbF(2r,r,n)$ equates to that of unitaries $(\bfB_i)_{i=1}^{n-1}$ in $\bbF^{r\times r}$ that satisfy~\eqref{eq.RH simplex}.
The latter equates to a simplex in a Radon--Hurwitz space:

\begin{theorem}
\label{thm.RH EITFF existence}
Let $r\geq 1$ and $n\geq 3$ be integers.
\begin{enumerate}
\renewcommand{\labelenumi}{(\alph{enumi})}
\item
Matrices $(\bfB_i)_{i=1}^{n-1}$ in $\bbF^{r\times r}$ are unitary and satisfy~\eqref{eq.RH simplex} if and only if they form a simplex in a $\rho$-space.
Defining $\bfPsi_{n-1}$ by~\eqref{eq.simplex recursion},
this occurs if and only if there exists a $\rho$-orthonormal sequence $(\bfC_i)_{i=1}^{n-2}$ in $\bbF^{r\times r}$ such that
\begin{equation}
\label{eq.RH simplex from RH orthonormal}
\bfB_j=\sum_{i=1}^{n-2}\bfPsi_{n-1}(i,j)\bfC_i,
\ \forall\,j\in[n-1].
\end{equation}

\item
An $\EITFF_\bbF(2r,r,n)$ exists if and only if $n\leq \rho_\bbF(r)+2$,
where $\rho_\bbF(r)$ is given by~\eqref{eq.RH number}.
\end{enumerate}
\end{theorem}

\begin{proof}
For (a), recall that vectors $(\bfphi_i)_{i=1}^{m}$ in a real space form a simplex if and only if they are unit norm and satisfy $\ip{\bfphi_i}{\bfphi_j}=-\frac1{m-1}$ for any distinct $i,j\in[m]$.
Thus, by~\eqref{eq.RH set relations},
members $(\bfB_i)_{i=1}^{n-1}$ of a $\rho$-space in $\bbF^{r\times r}$ form a simplex if and only if they are unitary and satisfy~\eqref{eq.RH simplex}.
Moreover, any unitaries $(\bfB_i)_{i=1}^{n-1}$ in $\bbF^{r\times r}$ satisfying~\eqref{eq.RH simplex} form a $\rho$-set,
and so their real span is a $\rho$-space.
The remainder of (a) then follows immediately from Lemma~\ref{lem.simplex}.

For (b), recall from Section~II that $\bbF^{r\times r}$ contains a \mbox{$\rho$-orthonormal} sequence $(\bfC_i)_{i=1}^{m}$ if and only if $m\leq\rho_\bbF(r)$.
In light of (a) and Theorem~\ref{thm.RH EITFF isometries},
an $\EITFF_\bbF(2r,r,n)$ thus exists if and only if $n-2\leq\rho_\bbF(r)$.
\end{proof}

Henceforth, we call any $\EITFF_\bbF(d,r,n)$ with $d=2r$ a \textit{Radon--Hurwitz EITFF},
and call $(\bfB_i)_{i=1}^{n-1}$ in $\bbF^{r\times r}$ a \mbox{\textit{$\rho$-simplex}} if its members are unitaries that satisfy~\eqref{eq.RH simplex}.
Applying Theorem~\ref{thm.RH EITFF existence}(a) to the explicit $\rho$-orthonormal sequences $(\bfC_i)_{i=1}^{m}$ given in Section~II---see Lemma~\ref{lem.complex RH} and subsequent discussion---yields explicit $\rho$-simplices of every possible size.
Applying Theorem~\ref{thm.RH EITFF isometries} to the latter thus yields Radon--Hurwitz EITFFs of every possible size.
For an instance of this process,
consider the following example.

\begin{example}
\label{ex.EITFF(4,2,4)}
By Theorem~\ref{thm.RH EITFF existence},
combining the $\rho$-orthonormal sequence $(\bfC_1,\bfC_2)=(\brI,\brR)$ in $\bbR^{2\times 2}$ from~\eqref{eq.def of T,M,R} via the matrix
\begin{equation*}
\bfPsi_3=\left[\begin{array}{ccc}
1&-\tfrac12&-\tfrac12\smallskip\\
0&\phantom{-}\frac{\sqrt{3}}2&-\frac{\sqrt{3}}2
\end{array}\right]
\end{equation*}
from~\eqref{eq.simplex recursion} yields a $\rho$-simplex
$(\bfB_1,\bfB_2,\bfB_3)$ in $\bbR^{2\times 2}$.
Here,
$\bfB_1=\brI$,
$\bfB_2=-\tfrac12\brI+\frac{\sqrt{3}}2\brR$ and
$\bfB_3=-\tfrac12\brI-\frac{\sqrt{3}}2\brR$.
Applying Theorem~\ref{thm.RH EITFF isometries} to it yields the following four explicit $4\times 2$ isometries of an $\EITFF_\bbR(4,2,4)$:
\begin{equation*}
\hspace{-1pt}\left[\begin{array}{c}\frac1{\sqrt{3}}\brI\smallskip\\\frac{\sqrt{2}}{\sqrt{3}}\brI\end{array}\right],
\left[\begin{array}{l}\phantom{-}\frac1{\sqrt{3}}\brI\smallskip\\-\tfrac1{\sqrt{6}}\brI+\tfrac1{\sqrt{2}}\brR\end{array}\right],
\left[\begin{array}{l}\phantom{-}\frac1{\sqrt{3}}\brI\smallskip\\-\tfrac1{\sqrt{6}}\brI-\tfrac1{\sqrt{2}}\brR\end{array}\right],
\left[\begin{array}{c}\brI\medskip\\\bfzero\end{array}\right].
\end{equation*}
\end{example}

\begin{remark}
Theorem~\ref{thm.RH EITFF existence} generalizes one portion of the seminal paper~\cite{LemmensS73b} on real equi-isoclinic subspaces by Lemmens and Seidel.
Though~\cite{LemmensS73b} predates ``tight fusion frame" terminology,
its Theorems~3.6 and~5.1 together imply that an $\EITFF_\bbR(2r,r,n)$ exists if and only if $n\leq\rho_\bbR(r)+2$.
From the perspective of this literature,
the novel contributions of Theorems~\ref{thm.RH EITFF isometries} and~\ref{thm.RH EITFF existence} are that they generalize these facts to the complex setting and,
unlike the Gram-matrix-based proof of Theorem~5.1 of~\cite{LemmensS73b},
they provide explicit isometries for any Radon--Hurwitz EITFF.
In general,
certain applications of EITFFs such as compressed sensing~\cite{EldarKB10,CalderbankTX15} rely on having explicit isometries for them,
and some of these applications such as radar~\cite{HermanS09} naturally arise in the complex setting.
\end{remark}

\begin{remark}
By Theorem~\ref{thm.RH EITFF existence},
an $\EITFF_\bbC(2r,r,n(r))$ only exists for an infinite number of $r$ if $n(r)$ is $\mathcal{O}(\log(r))$.
In contrast, Lemma~\ref{lem.Gerzon} only requires $n(r)\leq 3r^2+1$.
One might suspect that $\EITFF_\bbF(d,r,n)$ with large $r$ and $n$ are rare since all $r\binom{n}{2}$ of their principal angles must be equal.
Theorem~\ref{thm.RH EITFF existence} gives some formal credence to that intuition.
For context, the bound of Lemma~\ref{lem.Gerzon} is approached infinitely often when $r=1$ and $\bbF=\bbC$ since an $\EITFF_\bbC(q+1,1,q^2+q+1)$ exists for any prime power $q$~\cite{StrohmerH03,XiaZG05}.
\end{remark}

\begin{remark}
Theorems~\ref{thm.RH EITFF isometries} and~\ref{thm.RH EITFF existence} generalize to the setting where $\bbF$ is the quaternions $\bbH$;
see~\cite{AdamsLP66} for the appropriate extension of~\eqref{eq.RH number}.
For example, $(1,\rmi,\rmj,\rmk)$ is $\rho$-orthonormal in $\bbH$,
and applying Theorem~\ref{thm.RH EITFF isometries} with resulting $4$- and $5$-member \mbox{$\rho$-simplices} in $\bbH$ yields $5$- and $6$-vector ETFs in $\bbH^2$,
respectively;
ETFs with such parameters are known~\cite{CohnKM16,EtTaoui20,Waldron20}.
Hoggar's method~\cite{Hoggar77} converts any $\EITFF_\bbF(d,r,n)$ in which $\bbF$ is $\bbC$ or $\bbH$ into an $\EITFF_{\widehat{\bbF}}(2d,2r,n)$ where $\widehat{\bbF}$ is $\bbR$ or $\bbC$, respectively.
This method---which relates to the matrices of~\eqref{eq.def of T,M,R}, the Pauli spin matrices, and $\bbC^{2\times 2}$ representations of $\bbH$---converts any quaternionic Radon--Hurwitz EITFF into a complex one with a doubled ``$r$" parameter.
Since Theorems~\ref{thm.RH EITFF isometries} and~\ref{thm.RH EITFF existence} already give a full characterization of complex Radon--Hurwitz EITFFs, these are not novel.
Because of this, and the subtleties of noncommutative scalar multiplication~\cite{CohnKM16,EtTaoui20,Waldron20},
we leave quaternionic Radon--Hurwitz EITFFs to the interested reader.
\end{remark}

\renewcommand{\thesection}{\Roman{section}}
\section{Symmetries of Radon--Hurwitz EITFFs}
\renewcommand{\thesection}{\arabic{section}}

Let $d\geq r\geq 1$ be integers,
let $\calE$ be a $d$-dimensional space over $\bbF\in\set{\bbR,\bbC}$ and, for each $i\in[n]$,
let $\bfPhi_i$ be an isometry from $\bbF^r$ onto a subspace $\calU_i$ of $\calE$ with corresponding projection $\bfPi_i:=\bfPhi_i^{}\bfPhi_i^*$.
A permutation $\sigma$ of $[n]$ is a \textit{symmetry} of $(\calU_i)_{i=1}^{n}$ if $(\calU_{\sigma(i)})_{i=1}^{n}$ is equivalent to it in the sense of Lemma~\ref{lem.equivalence}, that is, if there exists a unitary $\bfUpsilon_{\sigma}:\calE\rightarrow\calE$ such that
\begin{equation}
\label{eq.symmetry subspace equivalence}
\bfPi_{\sigma(i)}=\bfUpsilon_{\sigma}^{}\bfPi_i\bfUpsilon_{\sigma}^*
\end{equation}
for all $i\in[n]$,
or equivalently, unitaries $(\bfZ_{\sigma,i})_{i=1}^n$ in $\bbF^{r\times r}$ such that
$\bfPhi_{\sigma(i)}^*\bfPhi_{\sigma(j)}^{}
=\bfZ_{\sigma,i}^*\bfPhi_i^*\bfPhi_j^{}\bfZ_{\sigma,j}^{}$
for all $i,j\in[n]$.
The symmetries of $(\calU_i)_{i=1}^{n}$ form a subgroup of the symmetric group $\rmS_n$ on $[n]$.
We say that $(\calU_i)_{i=1}^{n}$ has \textit{total} or \textit{alternating} symmetry when its symmetry group equals $\rmS_n$ or its alternating subgroup $\rmA_n$, respectively.
Equivalent sequences of subspaces have the same symmetry group.
When $(\calU_i)_{i=1}^{n}$ is a tight fusion frame for $\calE$ and $nr>d$,
its symmetry group equals that of its Naimark complements:
taking isometries \smash{$(\widetilde{\bfPhi}_i)_{i=1}^{n}$} that satisfy~\eqref{eq.Naimark},
\begin{align*}
\widetilde{\bfPhi}_{\sigma(i)}^*\widetilde{\bfPhi}_{\sigma(j)}^{}
&=-\tfrac{d}{nr-d}\bfPhi_{\sigma(i)}^*\bfPhi_{\sigma(j)}^{}\\
&=-\tfrac{d}{nr-d}\bfZ_{\sigma,i}^*\bfPhi_i^*\bfPhi_j^{}\bfZ_{\sigma,j}\\
&=\bfZ_{\sigma,i}^*\widetilde{\bfPhi}_i^*\widetilde{\bfPhi}_j^{}\bfZ_{\sigma,j},
\ \forall\,i,j\in[n],\, i\neq j.
\end{align*}

Recall that an $\EITFF_\bbF(d,r,n)$ is trivial if $d$ is $r$, $(n-1)r$ or $nr$.
All such EITFFs have total symmetry.
As a partial converse to this fact, any totally symmetric $\EITFF_\bbF(d,r,n)$ with $r=1$ is trivial in this sense~\cite{King19}.
That said, nontrivial EITFFs with total symmetry exist.
In fact, \cite{FickusIJM22,FickusIJM24} constructs an infinite number of them using classical representation theory,
including an $\EITFF_\bbR(5,2,5)$, cf.~\cite{EtTaoui06}.
Below, we give the second-ever construction of an infinite family of them, showing that many but not all Radon--Hurwitz EITFFs are totally symmetric.
Our proof relies on the explicit formulation of their isometries given in Theorem~\ref{thm.RH EITFF isometries}.
For context, none of the EITFFs of~\cite{FickusIJM22,FickusIJM24} are of Radon--Hurwitz ($d=2r$) type.

\begin{theorem}\
\label{thm.total symmetry suf}
\begin{enumerate}
\renewcommand{\labelenumi}{(\alph{enumi})}
\item
Let $(\bfPhi_i)_{i=1}^{n}$ be isometries of an $\EITFF_\bbF(2r,r,n)$ of the form given in Theorem~\ref{thm.RH EITFF isometries}.
This EITFF is totally symmetric if each $\bfB_i$ is skew-Hermitian,
i.e., $\bfB_i^*=-\bfB_i$.

\item
A skew-Hermitian $\rho$-simplex $(\bfB_i)_{i=1}^{n-1}$ in $\bbF^{r\times r}$ exists if and only if $n\leq \rho_\bbF(r)+1$, where $\rho_\bbF(r)$ is given by~\eqref{eq.RH number}.

\item
Any $\EITFF_\bbF(2r,r,n)$ has alternating or total symmetry.
\end{enumerate}
\end{theorem}

\begin{proof}
For (a),
define $\alpha,\beta$ by~\eqref{eq.RH EITFF parameters} and $(\bfPhi_i)_{i=1}^{n}$ by~\eqref{eq.RH EITFF isometries} where $(\bfB_i)_{i=1}^{n}$ is a sequence of skew-Hermitian unitaries in $\bbF^{r\times r}$ that satisfies~\eqref{eq.RH simplex}.
As such, $\bfB_i^2=-\brI$ for all $i\in[n-1]$, and
\begin{equation}
\label{eq.pf total sym 1}
\bfB_i\bfB_j+\bfB_j\bfB_i
=\tfrac{2}{n-2}\brI,\ \forall\,i,j\in[n-1],\,i\neq j.
\end{equation}
For any distinct $j,k\in[n-1]$, the matrix $\bfA_{j,k}:=\bfB_j-\bfB_k$ thus satisfies
$\bfA_{j,k}^*=-\bfA_{j,k}=\bfA_{k,j}$ and
\begin{equation*}
\bfA_{j,k}^2
=-2\brI-\bfB_j\bfB_k-\bfB_k\bfB_j
=-\tfrac1{\alpha^2}\brI.
\end{equation*}
Moreover, for any $i\in[n-1]$ with $i\notin\set{j,k}$,
\begin{align*}
\bfA_{j,k}\bfB_i
&=\bfB_j\bfB_i-\bfB_k\bfB_i\\
&=\tfrac{2}{n-2}\brI-\bfB_i\bfB_j-\tfrac{2}{n-2}\brI+\bfB_i\bfB_k\\
&=-\bfB_i\bfA_{j,k}.
\end{align*}
When $i$ is either $j$ or $k$, we instead have
\begin{align*}
\bfA_{j,k}\bfB_j
&=\bfB_j^2-\bfB_k\bfB_j
=\bfB_k^2-\bfB_k\bfB_j
=-\bfB_k\bfA_{j,k},\\
\bfA_{j,k}\bfB_k
&=\bfB_j\bfB_k-\bfB_k^2
=\bfB_j\bfB_k-\bfB_j^2
=-\bfB_j\bfA_{j,k}.
\end{align*}
It follows that $\alpha\bfA_{j,k}$ is a skew-Hermitian unitary for which
\begin{equation}
\label{eq.pf total sym 2}
(\alpha\bfA_{j,k})\bfB_i(\alpha\bfA_{j,k})
=\left\{\begin{array}{cl}
\bfB_k,&\ i=j,\\
\bfB_j,&\ i=k,\\
\bfB_i,&\ \text{else}.
\end{array}\right.
\end{equation}

For any $i\in[n]$,
the projection $\bfPi_i:=\bfPhi_i^{}\bfPhi_i^*$ arising from the isometry $\bfPhi_i$ defined by~\eqref{eq.RH EITFF isometries}
is
\begin{equation*}
\bfPi_i
=\left[\begin{array}{cc}
\alpha^2\brI&-\alpha\beta\bfB_i\\
\alpha\beta\bfB_i^{}&\beta^2\brI
\end{array}\right]\text{ if }
i<n,\
\bfPi_n
=\left[\begin{array}{cc}
\brI&\bfzero\\
\bfzero&\bfzero
\end{array}\right].
\end{equation*}
To show that this EITFF has total symmetry,
it suffices to take any $j,k\in[n]$ with $j<k$,
let $\sigma$ be the corresponding transposition $(jk)$,
and provide a corresponding unitary $\bfUpsilon_\sigma$ in $\bbF^{2r\times 2r}$ such that \eqref{eq.symmetry subspace equivalence} holds for all $i\in[n]$.
If $k<n$,
let
\begin{equation}
\label{eq.pf total sym 4}
\bfUpsilon_\sigma
=\alpha\left[\begin{array}{cc}
\bfB_j-\bfB_k&\bfzero\\
\bfzero&\bfB_k-\bfB_j
\end{array}\right]
=\alpha\left[\begin{array}{cc}
\bfA_{j,k}&\bfzero\\
\bfzero&-\bfA_{j,k}
\end{array}\right].
\end{equation}
To see that~\eqref{eq.symmetry subspace equivalence} holds for any $i\in[n-1]$,
note that by~\eqref{eq.pf total sym 2},
\begin{align*}
\bfUpsilon_\sigma^{}\bfPi_i^{}\bfUpsilon_\sigma^*
&=\left[\begin{array}{cc}
-\alpha^4\bfA_{j,k}^2&-\alpha^3\beta\bfA_{j,k}\bfB_i\bfA_{j,k}\\
\alpha^3\beta\bfA_{j,k}\bfB_i\bfA_{j,k}&-\alpha^2\beta^2\bfA_{j,k}^2
\end{array}\right]\\
&=\left[\begin{array}{cc}
\alpha^2\brI&-\alpha\beta\bfB_{\sigma(i)}\\
\alpha\beta\bfB_{\sigma(i)}&\beta^2\brI
\end{array}\right]\\
&=\bfPi_{\sigma(i)}.
\end{align*}
To see that~\eqref{eq.symmetry subspace equivalence} also holds when $i=n$,
note
\begin{equation*}
\bfUpsilon_\sigma^{}\bfPi_n^{}\bfUpsilon_\sigma^*
=\left[\begin{array}{cc}
-\alpha^2\bfA_{j,k}^2&\bfzero\\
\bfzero&\bfzero
\end{array}\right]
=\bfPi_n
=\bfPi_{\sigma(n)}.
\end{equation*}
If instead $k=n$,
let $\bfUpsilon_\sigma$ be the matrix
\begin{equation}
\label{eq.pf total sym 5}
\bfUpsilon_\sigma
=\left[\begin{array}{cc}
\alpha\bfB_j&\beta\brI\\
-\beta\brI&-\alpha\bfB_j
\end{array}\right].
\end{equation}
It is skew-Hermitian unitary since both $\bfUpsilon_\sigma^*=-\bfUpsilon_\sigma^{}$ and
\begin{equation*}
\bfUpsilon_\sigma^2
=\left[\begin{array}{cc}
\alpha^2\bfB_j^2-\beta^2\brI&\bfzero\\
\bfzero&-\beta^2\brI+\alpha^2\bfB_j^2
\end{array}\right]
=-\brI.
\end{equation*}
For this choice of $\sigma$ and $\bfUpsilon_\sigma$,
\eqref{eq.symmetry subspace equivalence} holds when $i=n$ since
\begin{align*}
\bfUpsilon_\sigma^{}\bfPi_n^{}\bfUpsilon_\sigma^*
&=\left[\begin{array}{cc}
\alpha\bfB_j&\beta\brI\\
-\beta\brI&-\alpha\bfB_j
\end{array}\right]
\left[\begin{array}{cc}
\brI&\bfzero\\
\bfzero&\bfzero
\end{array}\right]
\left[\begin{array}{cc}
-\alpha\bfB_j&-\beta\brI\\
\beta\brI&\alpha\bfB_j
\end{array}\right]\\
&=\left[\begin{array}{cc}
\alpha^2\brI&-\alpha\beta\bfB_j\\
\alpha\beta\bfB_j&\beta^2\brI
\end{array}\right]\\
&=\bfPi_j\\
&=\bfPi_{\sigma(n)}.
\end{align*}
Since $\bfUpsilon$ is a skew-Hermitian unitary,
this moreover implies that
$\bfUpsilon_\sigma^{}\bfPi_j^{}\bfUpsilon_\sigma^*
=\bfUpsilon_\sigma^*\bfPi_j^{}\bfUpsilon_\sigma^{}
=\bfPi_n
=\bfPi_{\sigma(j)}$,
and so~\eqref{eq.symmetry subspace equivalence} holds when $i=j$.
To see that~\eqref{eq.symmetry subspace equivalence} also holds for any $i\in[n-1]$ with $i\neq j$,
it suffices to show that
$\bfUpsilon_\sigma^{}\bfPi_i^{}=\bfPi_i^{}\bfUpsilon_\sigma^{}$.
Since
\begin{align*}
\bfUpsilon_\sigma^{}\bfPi_i^{}
=&\left[\begin{array}{cc}
\alpha^3\bfB_j+\alpha\beta^2\bfB_i&
-\alpha^2\beta\bfB_j\bfB_i+\beta^3\brI\\
-\alpha^2\beta\brI-\alpha^2\beta\bfB_j\bfB_i&
\alpha\beta^2\bfB_i-\alpha\beta^2\bfB_j
\end{array}\right],\\
\bfPi_i^{}\bfUpsilon_\sigma^{}
=&\left[\begin{array}{cc}
\alpha^3\bfB_j+\alpha\beta^2\bfB_i&
\alpha^2\beta\brI+\alpha^2\beta\bfB_i\bfB_j\\
\alpha^2\beta\bfB_i\bfB_j-\beta^3\brI&
\alpha\beta^2\bfB_i-\alpha\beta^2\bfB_j
\end{array}\right],
\end{align*}
this reduces to showing that both
\begin{align*}
-\alpha^2\beta\bfB_j\bfB_i+\beta^3\brI
&=\alpha^2\beta\brI+\alpha^2\beta\bfB_i\bfB_j,\\
-\alpha^2\beta\brI-\alpha^2\beta\bfB_j\bfB_i
&=\alpha^2\beta\bfB_i\bfB_j-\beta^3\brI,
\end{align*}
namely that $\bfB_i\bfB_j+\bfB_j\bfB_i=(\frac{\beta^2}{\alpha^2}-1)\brI$.
This is indeed the case, being a restatement of~\eqref{eq.pf total sym 1}.
Thus, the $\EITFF_\bbF(2r,r,n)$ arising from $(\bfPhi_i)_{i=1}^{n}$ is indeed totally symmetric.

For (b), let $(\bfB_i)_{i=1}^{n-1}$ be a skew-Hermitian $\rho$-simplex in $\bbF^{r\times r}$.
Its real span is a $\rho$-space in $\bbF^{r\times r}$ of dimension $n-2$ whose members are skew-Hermitian.
Let $(\bfC_i)_{i=1}^{n-2}$ be a \mbox{$\rho$-orthonormal} basis for this $\rho$-space and put $\bfC_{n-1}:=\brI$.
Since
$\bfC_{n-1}^*\bfC_i^{}+\bfC_i^*\bfC_{n-1}
=\bfC_i^{}+\bfC_i^*
=\bfzero$ for all $i\in[n-2]$,
the sequence $(\bfC_i)_{i=1}^{n-1}$ is $\rho$-orthonormal.
Thus, $n-1\leq\rho_\bbF(r)$.
Conversely, if $n\leq\rho_\bbF(r)+1$,
then there exists a \mbox{$\rho$-orthonormal} sequence $(\bfC_i)_{i=1}^{n-1}$ in $\bbF^{r\times r}$.
Here, $(\bfC_{n-1}^*\bfC_i^{})_{i=1}^{n-2}$ is a \mbox{$\rho$-orthonormal} sequence in $\bbF^{r\times r}$ whose real span is a $\rho$-space of dimension $n-2$ whose members are skew-Hermitian.
This space thus contains a skew-Hermitian $\rho$-simplex $(\bfB_i)_{i=1}^{n-1}$ in $\bbF^{r\times r}$.

For (c),
we show that any $\EITFF_\bbF(2\hat{r},\hat{r},n)$ has either alternating or total symmetry.
If $n\in\set{2,3}$, such an EITFF is trivial and has total symmetry.
Since any two equivalent EITFFs have the same symmetry group,
it thus suffices via Theorem~\ref{thm.RH EITFF isometries} to let $n\geq 4$ and consider EITFFs with isometries \smash{$(\widehat{\bfPhi}_i)_{i=1}^{n}$} in $\bbF^{2\hat{r}\times\hat{r}}$ of the form
\begin{equation*}
\widehat{\bfPhi}_i
=\left[\begin{array}{c}
\alpha\brI\\
\beta\widehat{\bfB}_i
\end{array}\right],
\ \forall\,i\in[n-1],
\quad
\widehat{\bfPhi}_n
=\left[\begin{array}{c}
\brI\\
\bfzero
\end{array}\right],
\end{equation*}
where $(\widehat{\bfB}_i)_{i=1}^{n-1}$ is a sequence of unitaries in
$\bbF^{\hat{r}\times\hat{r}}$ for which
\begin{equation}
\label{eq.pf total sym 6}
\widehat{\bfB}_i^*\widehat{\bfB}_j^{}+\widehat{\bfB}_j^*\widehat{\bfB}_i^{}
=-\tfrac{2}{n-2}\brI,
\ \forall\,i,j\in[n-1],\,i\neq j.
\end{equation}
Since $\rmA_n$ is a subgroup of $\rmS_n$ of index~$2$,
it moreover suffices to show that the symmetry group of this EITFF contains all products of two transpositions.
Now let $r=2\hat{r}$,
and define a sequence $(\bfB_i)_{i=1}^{n}$ of skew-Hermitian unitaries in $\bbF^{r\times r}$ by
\begin{equation*}
\bfB_i:=\left[\begin{smallmatrix}
\bfzero&-\widehat{\bfB}_i^*\\
\widehat{\bfB}_i^{}&\bfzero
\end{smallmatrix}\right].
\end{equation*}
In light of~\eqref{eq.RH set relations} and~\eqref{eq.pf total sym 6},
this sequence $(\bfB_i)_{i=1}^{n}$ satisfies \eqref{eq.pf total sym 1},
and so~\eqref{eq.RH EITFF isometries} defines isometries $(\bfPhi_i)_{i=1}^{n}$ of a totally symmetric $\EITFF_\bbF(2r,r,n)$ for $\bbF^{2r}=\bbF^{4\hat{r}}$.
Now consider the isometries $(\bfP\bfPhi_i)_{i=1}^{n}$ of the equivalent EITFF where
\begin{equation*}
\bfP:=\left[\begin{smallmatrix}
\brI&\bfzero&\bfzero&\bfzero\\
\bfzero&\bfzero&\bfzero&\brI\\
\bfzero&\brI&\bfzero&\bfzero\\
\bfzero&\bfzero&\brI&\bfzero
\end{smallmatrix}\right]\in\bbF^{4\hat{r}\times 4\hat{r}}
\end{equation*}
is a block permutation matrix,
namely
\begin{gather*}
\bfP\bfPhi_i
=\left[\begin{smallmatrix}
\alpha\brI&\phantom{-}\bfzero\\
\beta\widehat{\bfB}_i&\phantom{-}\bfzero\\
\bfzero&\phantom{-}\alpha\brI\\
\bfzero&-\beta\widehat{\bfB}_i^*
\end{smallmatrix}\right]\ \forall\,i\in[n-1],\quad
\bfP\bfPhi_n
=\left[\begin{smallmatrix}
\brI&\bfzero\\
\bfzero&\bfzero\\
\bfzero&\brI\\
\bfzero&\bfzero
\end{smallmatrix}\right].
\end{gather*}
Each associated projection
$\bfP\bfPi_i\bfP^*=\bfP\bfPhi_i^{}\bfPhi_i^*\bfP^*$
is a $4\hat{r}\times4\hat{r}$ block diagonal matrix whose upper left $2\hat{r}\times2\hat{r}$ corner
is the projection $\widehat{\bfPi}_i:=\widehat{\bfPhi}_i^{}\widehat{\bfPhi}_i^*$ onto the $i$th subspace of our arbitrary $\EITFF_\bbF(2\hat{r},\hat{r},n)$.
For any $\sigma\in\rmS_n$ and unitary $\bfUpsilon_\sigma$ in $\bbF^{2r\times 2r}$ that satisfy~\eqref{eq.symmetry subspace equivalence} for all $i\in[n]$,
\begin{equation*}
\bfP\bfPi_{\sigma(i)}\bfP^*
=(\bfP\bfUpsilon_\sigma\bfP^*)(\bfP\bfPi_i\bfP^*)(\bfP\bfUpsilon_\sigma\bfP^*)^*.
\end{equation*}
For instance, this holds when $\bfUpsilon_\sigma$ is given by~\eqref{eq.pf total sym 4} and $\sigma$ is the transposition $(jk)$ for some $j,k\in[n-1]$ with $j<k$; here,
\begin{align*}
\bfP\bfUpsilon_\sigma\bfP^*
&=\alpha\left[\begin{smallmatrix}
\bfzero&\bfzero&\widehat{\bfB}_k^*-\widehat{\bfB}_j^*&\bfzero\\
\bfzero&\bfzero&\bfzero&\widehat{\bfB}_k-\widehat{\bfB}_j\\
\widehat{\bfB}_j-\widehat{\bfB}_k&\bfzero&\bfzero&\bfzero\\
\bfzero&\widehat{\bfB}_j^*-\widehat{\bfB}_k^*&\bfzero&\bfzero
\end{smallmatrix}\right].
\end{align*}
It also holds when $\bfUpsilon_\sigma$ is given by~\eqref{eq.pf total sym 5} and
$\sigma$ is the transposition $(jn)$ for some $j\in[n-1]$; here,
\begin{align*}
\bfP\bfUpsilon_\sigma\bfP^*
&=\left[\begin{smallmatrix}
\bfzero&\bfzero&-\alpha\widehat{\bfB}_j^*&\beta\brI\\
\bfzero&\bfzero&\beta\brI&-\alpha\widehat{\bfB}_j\\
\alpha\widehat{\bfB}_j&\beta\brI&\bfzero&\bfzero\\
-\beta\brI&\alpha\widehat{\bfB}_j^*&\bfzero&\bfzero
\end{smallmatrix}\right].
\end{align*}
For any transpositions $\sigma_1$ and $\sigma_2$,
the product $(\bfP\bfUpsilon_{\sigma_1}\bfP^*)(\bfP\bfUpsilon_{\sigma_2}\bfP^*)$
of the associated unitaries is thus a $4\hat{r}\times4\hat{r}$ block diagonal matrix with diagonal blocks of size $2\hat{r}\times2\hat{r}$.
Since $(\bfP\bfUpsilon_{\sigma_1}\bfP^*)(\bfP\bfUpsilon_{\sigma_2}\bfP^*)$ conjugates
each $\bfP\bfPi_i\bfP^*$ to $\bfP\bfPi_{\sigma_1(\sigma_2(i))}\bfP^*$,
the upper left $2\hat{r}\times2\hat{r}$ corner of $(\bfP\bfUpsilon_{\sigma_1}\bfP^*)(\bfP\bfUpsilon_{\sigma_2}\bfP^*)$ conjugates each $\widehat{\bfPi}_i$ to $\widehat{\bfPi}_{\sigma_1(\sigma_2(i))}$.
As such, the product of any two transpositions in $\rmS_n$ indeed lies in the symmetry group of our arbitrary $\EITFF_\bbF(2\hat{r},\hat{r},n)$.
\end{proof}

By Theorem~\ref{thm.total symmetry suf},
an $\EITFF_\bbF(2r,r,n)$ with total symmetry exists if $n\leq\rho_\bbF(r)+1$.
Below, we show that when $\bbF=\bbC$,
this sufficient condition for the existence of a totally symmetric $\EITFF_\bbF(2r,r,n)$ is moreover necessary.
The situation is more complicated when $\bbF=\bbR$,
as some $\EITFF_\bbR(2r,r,n)$ with $n=\rho_\bbR(r)+2$ have total symmetry while others do not.
For example, $\rho_\bbR(8)=\rho_\bbC(8)=8$,
and no $\EITFF_\bbF(16,8,10)$ with $\bbF=\bbR$ has total symmetry since, more generally, no such EITFF with $\bbF=\bbC$ does.
In contrast,
$\rho_\bbR(2)=2$,
and the $\EITFF_\bbR(4,2,4)$ of Example~\ref{ex.EITFF(4,2,4)} has total symmetry:
its symmetry group contains $\rmA_4$ by Theorem~\ref{thm.total symmetry suf},
and moreover contains the transposition $\sigma=(23)$ since~\eqref{eq.symmetry subspace equivalence} holds for all $i\in[4]$ provided $\bfUpsilon_\sigma$ is taken to be
$[\begin{smallmatrix}\brM&\bfzero\\\bfzero&\brM\end{smallmatrix}]$
where $\brM=[\begin{smallmatrix}1&\phantom{-}0\\0&-1\end{smallmatrix}]$.
To better understand these facts,
consider the following technical characterization of the existence of $\EITFF_\bbF(2r,r,n)$ with total symmetry in terms of types of solutions to the classical Radon--Hurwitz equations~\eqref{eq.RH equations}.

\begin{lemma}
\label{lem.tot sym}
For any integers $r\geq1$ and $n\geq 4$,
an $\EITFF_\bbF(2r,r,n)$ with total symmetry exists if and only if there exist \mbox{$\rho$-orthonormal} matrices $(\bfC_i)_{i=1}^{n-2}$ in $\bbF^{r\times r}$ with $\bfC_1=\brI$ as well as a unitary $\bfU$ in $\bbF^{r\times r}$ that anticommutes with $\bfC_{n-2}$ and commutes with each $\bfC_i$ with $i\leq n-3$.
\end{lemma}

\begin{proof}
Assume that an $\EITFF_\bbF(2r,r,n)$ with total symmetry exists.
Since symmetry groups are preserved by equivalence,
and since multiplying every member of any $\rho$-simplex $(\bfB_i)_{i=1}^{n}$ in $\bbF^{r\times r}$ by $\bfB_1^*$ yields another $\rho$-simplex whose first member is $\brI$,
we assume without loss of generality via Theorem~\ref{thm.RH EITFF isometries} that it has isometries $(\bfPhi_i)_{i=1}^{n}$ of the form~\eqref{eq.RH EITFF isometries} where $(\bfB_i)_{i=1}^{n-1}$ is a $\rho$-simplex in $\bbF^{r\times r}$ with $\bfB_1=\brI$.
Then the transposition $\sigma$ that interchanges $n-2$ and $n-1$ is a symmetry of this EITFF,
and so Lemma~\ref{lem.equivalence} gives a unitary $\bfUpsilon_\sigma$ in $\bbF^{2r\times 2r}$ and unitaries $(\bfZ_{\sigma,i})_{i=1}^{n}$ in $\bbF^{r\times r}$ such that for any $i\in[n]$,
\begin{equation}
\label{eq.pf.lem.tot sym.1}
\bfPhi_{\sigma(i)}=\bfUpsilon_\sigma\bfPhi_i\bfZ_{\sigma,i}.
\end{equation}
Since
\smash{$\bfPhi_{\sigma(n)}
=\bfPhi_n
=[\begin{smallmatrix}\brI\\\bfzero\end{smallmatrix}]$},
applying~\eqref{eq.pf.lem.tot sym.1} when $i=n$ gives
\smash{$\bfUpsilon_\sigma=[\begin{smallmatrix}\bfU&\bfzero\\\bfzero&\bfV\end{smallmatrix}]$}
and $\bfZ_{\sigma,n}=\bfU^*$ for some unitaries $\bfU,\bfV$ in $\bbF^{r\times r}$.
For any $i\in[n-1]$,
\eqref{eq.pf.lem.tot sym.1} moreover guarantees that $\bfZ_{\sigma,i}=\bfU^*$ since both $\bfPhi_i$ and $\bfPhi_{\sigma(i)}$ have $\alpha\brI$ as their upper block.
Since $\bfZ_{\sigma,1}=\bfU^*$ and $\bfPhi_{\sigma(1)}=\bfPhi_1$ has $\beta\bfB_1=\beta\brI$ as its lower block,
applying~\eqref{eq.pf.lem.tot sym.1} when $i=1$ moreover gives that $\bfV=\bfU$.
Examining the bottom blocks of~\eqref{eq.pf.lem.tot sym.1} then gives
\begin{equation}
\label{eq.pf.lem.tot sym.2}
\bfU\bfB_i\bfU^*=\bfB_{\sigma(i)}\text{ when }2\leq i\leq n-1.
\end{equation}
In particular,
$\bfU$ commutes with each $\bfB_i$ with $i\in[n-3]$,
and anticommutes with $\bfB_{n-2}-\bfB_{n-1}$.
Now take $\rho$-orthonormal matrices $(\bfC_i)_{i=1}^{n-2}$ that satisfy~\eqref{eq.RH simplex from RH orthonormal}.
Here, Lemma~\ref{lem.simplex} gives that
$\bfC_1=\bfB_1=\brI$,
that each $\bfC_i$ lies in the span of $(\bfB_j)_{j=1}^{i}$ and so commutes with $\bfU$ for $i\leq n-3$,
and that $\bfC_{n-2}$ is a multiple of $\bfB_{n-2}-\bfB_{n-1}$ and so anticommutes with $\bfU$.

Conversely, if $(\bfC_i)_{i=1}^{n-2}$ is a $\rho$-orthonormal sequence in $\bbF^{r\times r}$ with $\bfC_1=\brI$ and there exists a unitary $\bfU$ in $\bbF^{r\times r}$ that both anticommutes with $\bfC_{n-2}$ and commutes with each $\bfC_i$ with $i\leq n-3$,
then~\eqref{eq.RH simplex from RH orthonormal} defines a $\rho$-simplex $(\bfB_i)_{i=1}^{n-1}$ in $\bbF^{r\times r}$ that satisfies~\eqref{eq.pf.lem.tot sym.2} where $\sigma$ is again the transposition that interchanges $n-2$ and $n-1$.
It follows that~\eqref{eq.pf.lem.tot sym.1} holds for all $i\in[n]$,
provided we define $(\bfPhi_i)_{i=1}^{n}$ by~\eqref{eq.RH EITFF isometries},
let $\bfUpsilon_\sigma$ be
$[\begin{smallmatrix}\bfU&\bfzero\\\bfzero&\bfU\end{smallmatrix}]$,
and let $\bfZ_{\sigma,i}$ be $\bfU^*$ for all $i\in[n]$.
As such, the symmetry group of the associated $\EITFF_\bbF(2r,r,n)$ contains $\sigma$.
Since it also contains $\rmA_n$ by Theorem~\ref{thm.total symmetry suf},
it is necessarily $\rmS_n$.
\end{proof}

For instance,
by Lemma~\ref{lem.tot sym}, the $\EITFF_\bbR(4,2,4)$ of Example~\ref{ex.EITFF(4,2,4)} has total symmetry since it arises from $(\bfC_1,\bfC_2)=(\brI,\brR)$ where $\brR$ anticommutes with the unitary $\bfU=\brM$.

\begin{theorem}
\label{thm.total symmetry complex}
For any integers $r\geq 1$ and $n\geq 3$,
an $\EITFF_\bbC(2r,r,n)$ with total symmetry exists if and only if $n\leq\rho_\bbC(r)+1$, where $\rho_\bbC(r)$ is given by~\eqref{eq.RH number}.
\end{theorem}

\begin{proof}
In light of Theorems~\ref{thm.RH EITFF existence} and~\ref{thm.total symmetry suf},
it suffices to assume that a totally symmetric $\EITFF_\bbC(2r,r,n)$ with $n=\rho_\bbC(r)+2$ exists and then arrive at a contradiction.
By Lemma~\ref{lem.tot sym}, there exists a $\rho$-orthonormal sequence $(\bfC_i)_{i=1}^{n-2}$ in $\bbC^{r\times r}$ with \mbox{$\bfC_1=\brI$} and a unitary $\bfU$ in $\bbC^{r\times r}$ that anticommutes with $\bfC_{n-2}$ and commutes with each $\bfC_i$ with $i\leq n-3$.

Since $\bbC$ is the underlying field,
there exists a unitary matrix $\bfV$ and a diagonal matrix $\bfD$,
both in $\bbC^{r\times r}$, such that $\bfU=\bfV\bfD\bfV^*$.
Here, without loss of generality, assume that the columns of $\bfV$ have been arranged so that common diagonal values of $\bfD$ appear consecutively,
and moreover that the same is true for $\bfD^2$.
To be precise, take distinct unimodular scalars $(z_k)_{k=1}^{l}$ and positive integers $(r_k)_{k=1}^l$ that sum to $r$ such that the $k$th diagonal block of $\bfD^2$ is $z_k\brI\in\bbC^{r_k\times r_k}$.
Replacing each $\bfC_i$ with $\bfV^*\bfC_i\bfV$ yields a $\rho$-orthonormal sequence $(\bfC_i)_{i=1}^{n-2}$ in $\bbC^{r\times r}$ with $\bfC_1=\brI$ where $\bfD$ anticommutes with $\bfC_{n-2}$ and commutes with each $\bfC_i$ with $i\leq n-3$.

Each matrix $\bfC_i$ commutes with $\bfD^2$,
and so is block diagonal,
with its $k$th diagonal block $\bfC_i^{(k)}$ belonging to $\bbC^{r_k\times r_k}$.
For any $k\in[l]$,
the sequence $(\bfC_i^{(k)})_{i=1}^{n-2}$ is thus \mbox{$\rho$-orthonormal},
implying $\rho_\bbC(r)=n-2\leq\rho_\bbC(r_k)$.
Moreover, there exists $k_0\in[l]$ such that $\rho_\bbC(r_{k_0})=\rho_\bbC(r)$:
otherwise,
$\bbC^{r_k\times r_k}$ contains $\rho_\bbC(r)+1$ $\rho$-orthonormal matrices for each $k\in[l]$,
and these can form the diagonal blocks of $\rho_\bbC(r)+1$ $\rho$-orthonormal matrices in $\bbC^{r\times r}$, a contradiction.
(We emphasize that these arguments remain valid in the degenerate case where $\bfU^2=z\brI$ for some $z\in\bbC$.
In that case, $\bfD^2=z\brI$,
implying that $l=1$, $r_1=r$ and $k_0=1$,
and moreover that each matrix $\bfC_i$ is block diagonal in a trivial way,
having $\bfC_i^{(1)}$ in $\bbC^{r\times r}$ as its sole block.)
Taking any such $k_0$,
replacing $r$ with $r_{k_0}$,
each $\bfC_i$ with $\bfC_i^{(k_0)}$,
and $\bfD$ with the appropriate unimodular multiple of its $k_0$th diagonal block yields a $\rho$-orthonormal sequence $(\bfC_i)_{i=1}^{n-2}$ in $\bbC^{r\times r}$ with $n=\rho_\bbC(r)+2$ and $\bfC_1=\brI$ where \smash{$\bfD=[\begin{smallmatrix}
\brI&\phantom{-}\bfzero\\
\bfzero&-\brI\end{smallmatrix}]$}
anticommutes with $\bfC_{n-2}$, and commutes with $\bfC_i$ for $i\leq n-3$.
Since $\bfD$ anticommutes with $\bfC_{n-2}$,
each of its two diagonal blocks lies in $\bbC^{\frac{r}{2}\times\frac{r}{2}}$ where
$r$ is even, and so $n=\rho_\bbC(r)+2=\rho_\bbC(\frac{r}{2})+4\geq 6$.
Moreover,
for any $i\in[n-3]$,
the fact that $\bfC_i$ commutes with $\bfD$ implies that it is of the form
\begin{equation*}
\bfC_i=\left[\begin{smallmatrix}
\bfC_i^{(1)}&\bfzero\\
\bfzero&\bfC_i^{(2)}
\end{smallmatrix}\right]
\end{equation*}
for some unitaries $\bfC_i^{(1)}$ and $\bfC_i^{(2)}$ in $\bbC^{\frac{r}{2}\times\frac{r}{2}}$.
Since $(\bfC_i)_{i=1}^{n-2}$ is $\rho$-orthonormal,
this implies that \smash{$(\bfC_i^{(1)})_{i=1}^{n-3}$} is as well,
and so $n-3\leq\rho_\bbC(\tfrac{r}{2})=\rho_\bbC(r)-2=n-4$, a contradiction.
\end{proof}

In the real setting,
Theorems~\ref{thm.RH EITFF existence} and~\ref{thm.total symmetry suf} imply that an $\EITFF_\bbR(2r,r,n)$ exists only if $n\leq\rho_\bbR(r)+2$,
and that a totally symmetric such EITFF exists if $n\leq\rho_\bbR(r)+1$.
The following result gives a partial characterization of the existence of totally symmetric such EITFFs in the remaining case.

\begin{theorem}
\label{thm.total symmetry real}
Let $\rho_\bbR(r)=8b+2^c$ where $r=(2a+1)2^{4b+c}$ for some nonnegative integers $a,b,c$ with $c\leq3$.
A totally symmetric $\EITFF_\bbR(2r,r,n)$ with $n=\rho_\bbR(n)+2$ exists if $c\in\set{0,1}$, and does not exist if $c=3$.
\end{theorem}

\begin{proof}
If $c=3$ then any totally symmetric $\EITFF_\bbR(2r,r,n)$ with $n=\rho_\bbR(r)+2$ would also be a totally symmetric $\EITFF_\bbC(2r,r,n)$ with $n=\rho_\bbC(r)+2$.
By Theorem~\ref{thm.total symmetry complex}, the latter does not exist.
If $c=0$ and $b=0$ then $n=\rho_\bbR(r)+2=8b+2^c+2=3$,
and a totally symmetric $\EITFF_\bbR(2r,r,3)$ exists,
namely the Naimark complement of the trivial $\EITFF_\bbR(r,r,3)$.
When instead either $c=1$ or $c=0$ and $b\geq 1$,
we have $n=\rho_\bbR(r)+2=8b+2^c+2\geq 4$,
implying by Lemma~\ref{lem.tot sym} that it suffices to construct a unitary $\bfU$ in $\bbR^{r\times r}$ as well as $m:=n-3=\rho_\bbR(r)-1=8b+2^c-1$ anticommuting skew-Hermitian matrices $(\bfC_i)_{i=1}^{m}$ in $\bbR^{r\times r}$ whose members commute with $\bfU$, save one, which anticommutes with $\bfU$.
Moreover, from any such $\bfU$ and $(\bfC_i)_{i=1}^{m}$ in $\bbR^{r\times r}$,
we can obtain analogous matrices in $\bbR^{\hat{r}\times\hat{r}}$ where $\hat{r}:=16r=(2a+1)2^{4(b+1)+c}$ by letting $\widehat{\bfU}:=\brI\otimes\bfU$ and taking $(\widehat{\bfC}_i)_{i=1}^{\hat{m}}$ to be the matrices of~\eqref{eq.real inflation},
where
\begin{equation*}
\hat{m}
:=m+8
=8(b+1)+2^c-1
=\rho_\bbR(\hat{r})-1.
\end{equation*}
As such, it suffices to consider just two cases,
namely when $(b,c)$ is either $(0,1)$ or $(1,0)$.
In the former, which equates to having $r=(2a+1)2^{4b+c}$ be an odd multiple of $2$,
let $\bfU$ be $\brM\otimes\brI$ and let $(\bfC_i)_{i=1}^{1}$ consist of the sole matrix $\brR\otimes\brI$.
In the latter, which equates to having $r=(2a+1)2^{4b+c}$ be an odd multiple of $16$,
let $\bfU$ and $(\bfC_i)_{i=1}^{8}$ be the tensor products of the identity matrix of size $2a+1$ with $\brI\otimes\brM\otimes\brM\otimes\brM$ and the matrices of~\eqref{eq.RH matrices of size 16}, respectively.
\end{proof}

In light of Theorems~\ref{thm.RH EITFF existence}, \ref{thm.total symmetry suf}, \ref{thm.total symmetry complex} and~\ref{thm.total symmetry real},
the existence of a totally symmetric $\EITFF_\bbF(2r,r,n)$ is settled except when $\bbF=\bbR$ and $(r,n)=(4(2a+1)16^b,8b+6)$ for some integers $a,b\geq0$.
The same technique we used to prove Theorem~\ref{thm.total symmetry real} implies that if such an EITFF exists for a particular value of $b$ then it also holds for all larger values of $b$.
Our preliminary work suggests that no $\EITFF_\bbR(8,4,6)$ has total symmetry.
We leave the remaining cases for future work.

\renewcommand{\thesection}{\Roman{section}}
\section{Conclusions and Future Work}
\renewcommand{\thesection}{\arabic{section}}

As part of their seminal paper on equi-isoclinic subspaces~\cite{LemmensS73b},
Lemmens and Seidel fully characterized the existence of real instances of such objects with $d=2r$ in terms of classical Radon--Hurwitz theory.
Refining and generalizing those ideas yields a full characterization of the existence of real and complex Radon--Hurwitz EITFFs,
as well as of the isometries thereof.
Each such EITFF has at least alternating symmetry,
and we have a near-complete characterization of when totally symmetric such EITFFs exist.
We highlight several avenues for follow-up research that were mentioned above:
finish characterizing the existence of real Radon--Hurwitz EITFFs with total symmetry;
develop a general theory of totally symmetric EITFFs that unites the constructions here with those of~\cite{FickusIJM22,FickusIJM24};
apply totally symmetric EITFFs to compressed sensing problems.

In contrast to several (``symmetry implies optimality") results from the recent literature~\cite{IversonM22,IversonM24},
which exploit the symmetries of certain sequences of subspaces to show that they are EITFFs,
Theorem~\ref{thm.total symmetry suf} is notably of the converse (``optimality implies symmetry") type, showing that any Radon--Hurwitz EITFF has at least alternating symmetry.
That said, both Theorem II.6 of~\cite{King19},
which implies that every triply-transitive $\EITFF_\bbF(d,1,n)$ is trivial,
and Theorem~\ref{thm.total symmetry complex} show that too much symmetry can be prohibitive.

\section*{Acknowledgments}
\noindent
We thank the editors and the four anonymous reviewers for their comments and suggestions,
all of which were helpful.
This research was supported in part by the Air Force Institute of Technology through the Air Force Office of Scientific Research Summer Faculty Fellowship Program, Contract Numbers FA8750-15-3-6003, FA9550-15-0001 and FA9550-20-F-0005.
JWI was partially supported by NSF DMS 2220301 and a grant from the Simons Foundation.
The views expressed in this article are those of the authors and do not reflect the official policy or position of the United States Air Force, Department of Defense, or the U.S.~Government.

\appendix

\begin{proof}[Proof of Lemma~\ref{lem.equivalence}]
Clearly, (i) and (ii) are equivalent, and (iii) implies both (ii) and (iv).
It thus suffices to prove that each of (ii) and (iv) implies (iii).
If (ii) holds then for any $i\in[n]$, $\widehat{\bfPhi}_i^{}\widehat{\bfPhi}_i^*=(\bfUpsilon\bfPhi_i^{})(\bfUpsilon\bfPhi_i^{})^*$,
and so $\widehat{\bfPhi}_i=\bfUpsilon\bfPhi_i\bfZ_i$ for some unitary $\bfZ_i$ in $\bbF^{r\times r}$, implying (iii).
If (iv) holds, the fusion Gram matrix associated to $(\widehat{\bfPhi}_i)_{i=1}^{n}$ equals that of $(\bfPhi_i\bfZ_i)_{i=1}^{n}$,
and so $\widehat{\bfPhi}^*\widehat{\bfPhi}=\bfZ^*\bfPhi^*\bfPhi\bfZ$ where $\bfZ$ is the block diagonal matrix whose $i$th diagonal block is $\bfZ_i$.
As such, $\widehat{\bfPhi}=\bfUpsilon_0\bfPhi\bfZ$ for some unitary map $\bfUpsilon_0$ from $\im(\bfPhi\bfZ)$ to $\im(\widehat{\bfPhi})$,
namely from the span of $(\calU_i)_{i=1}^{n}$ to that of $(\widehat{\calU}_i)_{i=1}^{n}$.
Since $\calE$ and $\widehat{\calE}$ have equal dimension,
$\bfUpsilon_0$ extends to a unitary map \mbox{$\bfUpsilon:\calE\rightarrow\widehat{\calE}$}.
Thus, $\widehat{\bfPhi}=\bfUpsilon\bfPhi\bfZ$.
Moreover, for any $i\in[n]$, applying $\widehat{\bfPhi}=\bfUpsilon\bfPhi\bfZ$ to $(\bfdelta_{i,k})_{k=1}^r$ gives $\widehat{\bfPhi}_i=\bfUpsilon\bfPhi_i\bfZ_i$ and so (iii).
\end{proof}

\begin{proof}[Proof of Lemma~\ref{lem.Gerzon}]
The sequence of projections $(\bfPi_i)_{i=1}^{n}$ of such subspaces satisfies~\eqref{eq.equi-isoclinic projections} where $0\leq\sigma<1$.
It lies in the real space $\SA(\calE)$ of self-adjoint operators on $\calE$.
Since
$\Tr(\bfPi_i\bfPi_j)=\Tr(\bfPi_i\bfPi_j\bfPi_i)=\Tr(\sigma^2\bfPi_i)=r\sigma^2$
whenever $i\neq j$,
its Gram matrix with respect to the Frobenius inner product is  $r[(1-\sigma^2)\brI+\sigma^2\brJ]$.
Since $\sigma<1$,
this matrix is invertible,
and so $(\bfPi_i)_{i=1}^{n}$ is linearly independent.
It thus suffices to show that the right-hand side of~\eqref{eq.Gerzon bound} is the dimension of some subspace $\calU$ of $\SA(\calE)$ that contains $(\bfPi_i)_{i=1}^{n}$.
Since~\eqref{eq.equi-isoclinic projections} holds and $\bfPi_1^3=\bfPi_1^{}$,
we can take $\calU$ to be the space of all self-adjoint $\bfA:\calE\rightarrow\calE$ such that $\bfPi_1\bfA\bfPi_1$ is a real scalar multiple of $\bfA$.
Without loss of generality,
Lemma~\ref{lem.equivalence} gives $\calE=\bbF^d$ and $\bfPi_1=[\begin{smallmatrix}\brI&\bfzero\\\bfzero&\bfzero\end{smallmatrix}]$.
Thus, $\calU$ consists of all $\bfA\in\bbF^{d\times d}$ of the form
\begin{equation*}
\bfA=\left[\begin{smallmatrix}
x\brI&\bfA_{2,1}^*\\
\bfA_{2,1}&\bfA_{2,2}
\end{smallmatrix}\right]
\end{equation*}
where $x\in\bbR$, $\bfA_{2,1}\in\bbF^{(d-r)\times r}$ and $\bfA_{2,2}^*=\bfA_{2,2}$.
Computing the dimension of this real space $\calU$ gives the result.
\end{proof}

\begin{proof}[Proof of Lemma~\ref{lem.complex RH}]
For any integers $r_1,r_2\geq1$,
let $r=r_1+r_2$ and let $\bfD_{r_1,r_2}$ be the block diagonal matrix in $\bbF^{r\times r}$ whose first and second diagonal blocks are $\brI\in\bbF^{r_1\times r_1}$ and $-\brI\in\bbF^{r_2\times r_2}$, respectively.
Write any $\bfA\in\bbF^{r\times r}$ blockwise as
\begin{equation*}
\bfA=\left[\begin{smallmatrix}
\bfA_{1,1}&\bfA_{1,2}\\
\bfA_{2,1}&\bfA_{2,2}
\end{smallmatrix}\right],
\text{ where }
\bfA_{2,1}\in\bbF^{r_2\times r_1},\,
\bfA_{1,2}\in\bbF^{r_1\times r_2}.
\end{equation*}
If $\bfA$ is unitary, then it commutes with $\bfD_{r_1,r_2}$ if and only if $\bfA_{1,1}$ and $\bfA_{2,2}$ are unitary, and it anticommutes with $\bfD_{r_1,r_2}$ if and only if $\bfA_{2,1}$ and $\bfA_{1,2}$ are unitary.
The latter can only occur if $\bfA_{2,1}$ is square, in which case $r=2r_1=2r_2$ is even and
$\bfD_{r_1,r_2}=\brM\otimes\brI$ where
$\brM=[\begin{smallmatrix}1&\phantom{-}0\\0&-1\end{smallmatrix}]$.

In particular, if $(\bfC_i)_{i=1}^{m}$ in $\bbF^{r\times r}$ is $\rho$-orthonormal,
then replacing each $\bfC_i$ with $\bfC_m^*\bfC_i^{}$ yields an equivalent such sequence with $\bfC_m=\brI$.
If $\bbF=\bbC$ and $m\geq2$,
then diagonalizing the skew-Hermitian matrix $\bfC_{m-1}$ as $\bfC_{m-1}=\bfV(\rmi\bfD_{r_1,r_2})\bfV^*$ for some unitary $\bfV$,
and replacing each $\bfC_i$ with $\bfV^*\bfC_i\bfV$ yields an equivalent $\rho$-orthonormal sequence $(\bfC_i)_{i=1}^{m}$ in $\bbC^{r\times r}$ with \mbox{$\bfC_m=\brI$} and $\bfC_{m-1}=\rmi\bfD_{r_1,r_2}$.
If moreover $m\geq 3$ then each $\bfC_i$ with $i\leq m-2$ anticommutes with $\bfD_{r_1,r_2}$,
implying $r=2r_1=2r_2$ and that each of the two $\frac{r}{2}\times\frac{r}{2}$ diagonal blocks of each $\bfC_i$ with $i\leq m-2$ is zero.
Since $(\bfC_i)_{i=1}^{m-2}$ is skew-Hermitian and $\rho$-orthonormal,
its members have form~\eqref{eq.complex reduction} for some $\rho$-orthonormal sequence $(\widehat{\bfC}_i)_{i=1}^{m-2}$ in $\bbC^{\frac{r}{2}\times\frac{r}{2}}$.
Conversely,
in light of~\eqref{eq.RH set relations},
it is straightforward to verify that
any sequence $(\bfC_i)_{i=1}^{m}$ in $\bbC^{r\times r}$ of this form is $\rho$-orthonormal.
\end{proof}

\begin{proof}[Proof of Lemma~\ref{lem.simplex}]
Let $(\bfphi_i)_{i=1}^{m}$ be a sequence of $m\geq 2$ vectors in a real space,
and let $\calE$ and $\bfPhi:\bbR^m\rightarrow\calE$ be its span and synthesis map, respectively.
If $\bfUpsilon:\bbR^{m-1}\rightarrow\calE$ is the synthesis map of a sequence $(\bfupsilon_i)_{i=1}^{m-1}$ in $\calE$, then
\begin{equation}
\label{eq.pf.lem.simplex.1}
\bfUpsilon\bfPsi_m\bfdelta_j
=\bfUpsilon\sum_{i=1}^{m-1}\bfPsi_m(i,j)\bfdelta_i
=\sum_{i=1}^{m-1}\bfPsi_m(i,j)\bfupsilon_i
\end{equation}
for each $j\in[m]$.
In particular, if $(\bfupsilon_i)_{i=1}^{m-1}$ is an orthonormal basis for $\calE$ that satisfies~\eqref{eq.simplex basis},
then $\bfUpsilon$ is unitary and $\bfUpsilon\bfPsi_m\bfdelta_j$ equals $\bfphi_j=\bfPhi\bfdelta_j$ for each $j\in[m]$.
In this case, $\bfPhi=\bfUpsilon\bfPsi_m$, implying $\bfPhi^*\bfPhi=\bfPsi_m^*\bfPsi_m^{}=\tfrac1{m-1}(m\brI-\brJ)$,
and so $(\bfphi_i)_{i=1}^{m}$ is a simplex.
Conversely, if $(\bfphi_i)_{i=1}^{m}$ is a simplex then its Gram matrix is $\bfPhi^*\bfPhi=\tfrac1{m-1}(m\brI-\brJ)=\bfPsi_m^*\bfPsi_m^{}$,
implying \mbox{$\bfPhi=\bfUpsilon\bfPsi_m$} for some unitary  $\bfUpsilon:\bbR^{m-1}\rightarrow\calE$.
When this occurs, $(\bfupsilon_i)_{i=1}^{m-1}:=(\bfUpsilon\bfdelta_i)_{i=1}^{m-1}$ is an orthonormal basis for $\calE$ that has $\bfUpsilon$ as its synthesis map,
and~\eqref{eq.simplex basis} follows from the fact that~\eqref{eq.pf.lem.simplex.1} equals $\bfPhi\bfdelta_j=\bfphi_j$ for each $j\in[m]$.

The remaining claims follow from~\eqref{eq.simplex basis} and properties of $\bfPsi_m$.
By induction, $\bfPsi_m$ is an upper triangular matrix whose diagonal entries are positive and whose first entry is $1$.
As such, the $j=1$ instance of~\eqref{eq.simplex basis} gives
$\bfphi_1
=\bfPsi_m(1,1)\bfupsilon_1
=\bfupsilon_1$,
namely (a).
Next, since the matrix $\bfPsi_{m,0}$ obtained by removing the final column of $\bfPsi_m$ is invertible,
we have $\bfUpsilon=\bfPhi_0^{}\bfPsi_{m,0}^{-1}$ where $\bfPhi_0$ is the synthesis map of $(\bfphi_i)_{i=1}^{m-1}$ and $\bfPsi_{m,0}^{-1}$ is upper triangular,
implying
$\bfupsilon_j
=\sum_{i=1}^{j}\bfPsi_{m,0}^{-1}(i,j)\bfphi_i
\in\Span(\bfphi_i)_{i=1}^{j}$
for any $j\in[m-1]$, namely (b).
Finally, by induction, $\bfPsi_m(i,m-1)=\bfPsi_m(i,m)$ for all $i\in[m-2]$ while
$0<\bfPsi_m(m-1,m-1)=-\bfPsi_m(m-1,m)$,
and combining these facts with~\eqref{eq.simplex basis} gives (c).
\end{proof}

\vskip 0pt plus -1fil

\begin{IEEEbiographynophoto}{Matthew Fickus} (M'08--SM'18) earned a Ph.D.\ in Mathematics from the University of Maryland in 2001.  In 2004, he joined the Department of Mathematics and Statistics, Air Force Institute of Technology, where he is currently a Professor of Mathematics.  His research interests include applying harmonic analysis and combinatorial design to problems in information theory and signal processing.
\end{IEEEbiographynophoto}

\vskip 0pt plus -1fil

\begin{IEEEbiographynophoto}{Enrique Gomez-Leos} earned a Master of Mathematical Sciences degree at The Ohio State University in 2020.  He is currently a 5th year Ph.D.\ student in the Department of Mathematics at Iowa State University. His research interests include extremal and random graph theory.
\end{IEEEbiographynophoto}

\vskip 0pt plus -1fil

\begin{IEEEbiographynophoto}{Joseph W.~Iverson} graduated from the University of Oregon with a Ph.D.\ in Mathematics in 2016.  Since 2018, he has been a member of the Department of Mathematics at Iowa State University, where he is now an Associate Professor.  His research interests include symmetry and its applications in harmonic analysis.
\end{IEEEbiographynophoto}

\end{document}